\documentclass[10pt,journal,a4paper]{IEEEtran}
\usepackage{amssymb,amsmath}
\usepackage{cite}
\usepackage{psfrag}
\usepackage{url}
\usepackage[latin1]{inputenc}
\usepackage[absolute,overlay]{textpos}
\newcommand{\argmax}[1]{{\operatorname{arg}\,\max_{#1}}\,}
\newcommand{\argmin}[1]{{\operatorname{arg}\,\min_{#1}}\,}
\usepackage{pgf}
\usepackage[latin1]{inputenc}
\usepackage[export]{adjustbox}
\usepackage{hhline}
\usepackage{multirow}
\usepackage{arydshln}
\usepackage{bm}
\usepackage{bbm}
\usepackage{booktabs}

\usepackage{amsthm}

\newtheorem{theorem}{Theorem}

\newtheorem{remark}{Remark}

\begin{document}

\title{On CSI-free Multi-Antenna Schemes for Massive RF Wireless Energy Transfer}
\author{
	\IEEEauthorblockN{Onel L. A. L\'opez, 
		Samuel Montejo-S\'anchez,
		Richard D. Souza, 
		Constantinos B. Papadias,
		Hirley Alves
	}
	\thanks{Onel L. A. L\'opez and Hirley Alves are with the Centre for Wireless Communications (CWC), University of Oulu, Finland. \{onel.alcarazlopez,hirley.alves\}@oulu.fi}
	\thanks{S. Montejo-S\'anchez is with Programa Institucional de Fomento a la I+D+i, Universidad Tecnol\'ogica Metropolitana, Santiago, Chile. \{smontejo@utem.cl\}.}
	\thanks{Richard D. Souza is with Federal University of Santa Catarina (UFSC), Florian\'opolis, Brazil. \{richard.demo@ufsc.br\}.}
	\thanks{Constantinos B. Papadias is with the Research, Technology \& Innovation Network (RTIN) at The American College of Greece, Athens 15342, Greece. \{cpapadias@acg.edu\}.}
	\thanks{This work is supported by Academy of Finland (Aka) (Grants n.307492, n.318927 (6Genesis Flagship), n.319008 (EE-IoT)), as well as Finnish Foundation for Technology Promotion, FIREMAN (Grant n.326301), FONDECYT Postdoctoral Grant n.3170021 in Chile, the National Council for Scientific and Technological Development (CNPq) and by project Print CAPES-UFSC ``Automation 4.0'' in Brazil.}
	\thanks{\textcopyright 2020 IEEE. Accepted for publication in the Internet of Things Journal. Personal use of this material is permitted. Permission from IEEE must be obtained for all other uses, in any current or future media, including reprinting/republishing this material for advertising or promotional purposes, creating new collective works, for resale or redistribution to servers or lists, or reuse of any copyrighted component of this work in other works.}
}
					
\maketitle

\begin{abstract}
Radio Frequency Wireless Energy Transfer (RF-WET) is emerging as a potential green enabler for massive Internet of Things (IoT). Herein, we analyze Channel State Information (CSI)-free multi-antenna strategies for powering wirelessly a large set of single-antenna IoT devices. The CSI-free schemes are $\mathrm{AA-SS}$ ($\mathrm{AA-IS}$), where all antennas transmit the same (independent) signal(s), and $\mathrm{SA}$, where just one antenna transmits at a time such that all  antennas are utilized during the coherence block.
We characterize the distribution of the provided energy under correlated Rician fading for each scheme and find out that while $\mathrm{AA-IS}$ and $\mathrm{SA}$ cannot take advantage of the multiple antennas to improve the average provided energy, its dispersion can be significantly reduced. 
Meanwhile, $\mathrm{AA-SS}$ provides the greatest average energy, but also the greatest energy dispersion, and the gains depend critically on the mean phase shifts between the antenna elements.
We find that consecutive antennas must be $\pi$ phase-shifted for optimum average energy performance under $\mathrm{AA-SS}$.
Our numerical results evidence that correlation is beneficial under $\mathrm{AA-SS}$, while a greater line of sight (LOS) and/or number of antennas is not always beneficial under such scheme. 	 
Meanwhile, both $\mathrm{AA-IS}$ and $\mathrm{SA}$ schemes benefit from small correlation, large LOS and/or large number of antennas.
Finally, $\mathrm{AA-SS}$ ($\mathrm{SA}$ and $\mathrm{AA-IS}$) is (are) preferable when devices are (are not) clustered in specific spatial directions.
\end{abstract}
\vspace{-2mm}
\section{Introduction}
Internet of Things (IoT), where every \textit{thing} is practically transformed into an information source, represents a major technology trend that is  revolutionizing the way we interact with our surrounding environment so that we can make the most of it. There are two general categories of IoT use cases \cite{Shirvanimoghaddam.2017}: i) critical IoT, with stringent requirements on reliability, availability, and low latency, e.g. in remote health care, traffic safety and control, industrial applications and control, remote manufacturing, surgery; and ii) massive IoT, where  sensors typically report to the cloud, and the requirement is for low-cost devices with low energy consumption and good coverage. 
In such massive deployments, these IoT nodes are not generally supposed to be transmitting continuously, but they do require to operate for long periods of time without batteries replacement, especially since many of them could be placed in  hazardous environments, building structures or the human body.
One important enabler under consideration is Wireless Energy Transfer (WET). 
Notice that WET is a general concept that includes several power transfer technologies such as those
employing for instance ultra-sound, inductive, capacitive, or resonant coupling \cite{Wang.2020}. In this work we focus on Radio Frequency (RF)-based WET which allows the IoT nodes equipped with an energy harvesting (EH) circuitry  to harvest energy from incoming RF signals \cite{Chae.2018,Hou.2018,Saad.2019,Lopez.2020}. Hereinafter, we refer to RF-WET just as WET.

WET holds vast potential for replacing batteries or increasing their lifespans. 
In fact, RF-EH devices can become self-sustaining  with respect to the energy required for operation, thereby  obtaining an unlimited operating lifespan while demanding negligible maintenance \cite{Tran.2017}. This is crucial for the future society since the battery waste processing is already a critical problem. The most effective approach for reducing battery waste is to avoid using them, for which WET is an attractive clean solution. 
Notice that with the continuous advances on circuitry technology aiming at reducing further the power consumption of low-cost devices\footnote{Such current technological advances point to power consumption in the order of nW and $\mu$W (or even less) in sleep and active modes, respectively.
Nowadays, there is available a wide-range of IoT devices that seem suitable for relying on an RF-EH module as the main power supply due to their extremely low power consumption profiles. For instance, accelerometer ADXL362 ($\sim 40\mu$W in active) and Light ISL29033 ($\sim 100\mu$W in active) \cite{Jayakumar.2014}. In fact, the first wireless battery-free bio-signal processing system on chip was introduced by \cite{Zhang.2013} and it is able of monitoring various bio-signals via electrocardiogram, electromyogram, and electroencephalogram. The total size of the chip is 8.25$\mathrm{mm}^2$ and consumes 19$\mu$W to  measure the heart rate.}, 
e.g. advances on Complementary Metal-Oxide-Semiconductor (CMOS) and Micro-Electro-Mechanical Systems (MEMS) technologies \cite{Zhang.2018}, energy-efficient Proportional to the Absolute Temperature (PTAT) circuit implementations \cite{Alazzawi.2019}, and printed sensors technology \cite{Moreno.2019}, an exponential growth on WET-enabled IoT applications is expected.

The IoT paradigm intrinsically includes wireless information transfer (WIT), thus
WET appears naturally combined with WIT. In such case two main architectures can be distinguished in literature: i) Wireless Powered Communication Network (WPCN), where WET occurs in the downlink in a first phase and WIT takes place in the second phase; and ii) Simultaneous Wireless Information and Power Transfer (SWIPT), where WET and WIT occur simultaneously. 
Readers can refer to \cite{Clerckx.2019} to review the recent progress on both architectures, while herein the discussions will focus merely on WPCN and pure WET setups. Notice that in most of practical applications WET duration would be significantly larger than WIT in order to harvest usable amounts of energy \cite{Lopez.2020}. Actually, some use cases require operating under WET almost permanently while WIT happens sporadically, e.g. due to event-driven traffic. Therefore, enabling efficient WET is mandatory for realizing the IoT paradigm and constitutes the scope of this work.
\vspace{-2mm}
\subsection{Related Work}\label{RW}
Recent works have specifically considered  WET and WPCN setups in different contexts and scenarios. 
Key networking structures and performance enhancing techniques to build an efficient WPCN are discussed in \cite{Bi.2016}, where authors also point out  challenging research directions. 
Departing from the simple Harvest-then-Transmit (HTT) scheme \cite{Ju.2014,Lopez.2017,LopezAlves.2019} several other protocols have been proposed over the past few years to boost the WPCN performance such as the Harvest-then-Cooperate (HTC) system studied in  \cite{ChenLi.2015,Lopez.2018,LopezDemo.2017} and the power control scheme relying on energy accumulation between transmission rounds discussed in \cite{LopezE.2018}.
Authors either analyze the performance of the information transmission phase, or optimize it by using power control or cooperative schemes. Some scheduling strategies that allow a direct optimization of the energy efficiency of the network are also proposed in \cite{Bacinoglu.2018,Yang.2018}.
Additionally, an energy cooperation scheme that enables energy cooperation in battery-free wireless networks with WET is presented in \cite{Li.2018}.
Meanwhile, the deployment of single-antenna power beacons (PBs) for powering the mobiles in an uplink cellular network is proposed in \cite{Huang.2014}. Therein, authors investigate the network performance under an outage constraint on data links using stochastic-geometry tools, while they corroborate the effectiveness of relying on directed WET instead of using isotropic antennas.

Yet, shifts in the system architecture and in the resource allocation strategies for optimizing the energy supply to massive IoT deployments are still required. In \cite{Lopez.2020} we discuss several techniques that seem suitable
for enabling WET as an efficient solution for powering the future IoT networks. They are:
\begin{itemize}
	\item \textbf{Energy beamforming (EB)}, which allows the energy signals at different antennas to be carefully weighted to achieve constructive superposition at intended receivers. The larger the number of antennas installed at the PB,
	the sharper the energy beams can be generated in some particular spatial directions. The EB benefits for WPCNs have been investigated for instance in \cite{Huang.2016} in terms of average throughput performance, while in \cite{Park.2017} authors propose an EB scheme that maximizes the weighted sum of the harvested energy and 	the information rate in a multiple-input single-output (MISO) system. However, the benefits of EB in practice depend on the available
	Channel State Information (CSI) at the transmitter, and although there has been some works proposing adequate channel acquisition methods, e.g. \cite{Zeng.2015,ZengZhang.2015}, this still constitutes a serious limitation. This is due to the harsh requirements in terms of energy and scheduling policies, which become even more critical as the number of EH devices increases;
	\item \textbf{Distributed Antenna Systems (DAS)}, which are capable of eliminating blind spots while homogenizing the energy provided to a given area and supporting ubiquitous energy accessibility. The placement optimization of single-antenna energy and information access points in WPCNs is investigated in \cite{BiZhang.2016}, where authors focus on minimizing the network deployment cost subject to energy harvesting and communication performance constraints. On the other hand,  authors in \cite{Chen.2017} study the probability density function (PDF), the cumulative distribution function (CDF), and the average of the energy harvested in DAS, while they determine appropriate strategies when operating under different channel conditions by using such information. Although works in this regard have avoided the use of multiple transmit antennas, we would like to highlight the fact that multiple separate PBs, each equipped with multiple transmit antennas, could alleviate the issue of CSI acquisition when 	forming efficient energy beams in multiple-users setups, since each PB may be responsible for the CSI acquisition procedure of a smaller set of EH devices;
	\item \textbf{CSI-limited/CSI-free schemes}. Even without accounting for the considerable energy resources demanded by CSI acquisition, the performance of CSI-based systems decays quickly as the number of served devices increases. Therefore, in massive deployment scenarios the broadcast nature of wireless transmissions should be intelligently exploited for powering simultaneously a massive number of IoT devices  with minimum or non CSI. For instance, authors of \cite{Clerckx.2018} propose a method that relies on multiple dumb antennas transmitting phase-shifted signals to induce fast fluctuations on a slow-fading wireless channel and attain transmit diversity. Also,  we recently analyzed in \cite{Lopez.2019} several CSI-free multi-antenna schemes that a PB may utilize to efficiently power  a large set of nearby EH devices, while we discussed their performance in Rician correlated fading channels. We found out that i) the switching antenna ($\mathrm{SA}$) strategy, where a single antenna  transmits at a time with full power, provides the most predictable energy source, and it is particularly suitable for powering sensor nodes with highly sensitive EH hardware operating under non line of sight (NLOS); and ii) transmitting simultaneously the same signal with equal power in all antennas ($\mathrm{AA}$, but herein referred as $\mathrm{AA-SS}$) is the most beneficial scheme when LOS increases and it is the only scheme that benefits from spatial correlation. 
	Notice that $\mathrm{SA}$ and $\mathrm{AA-SS}$ are respectively special cases of \textit{unitary} and \textit{uniform query} schemes  proposed in the context of backscattering communications \cite{He.2015}. In fact, the authors of \cite{He.2015,He.2016} show that \textit{unitary query} can provide considerable performance gains with respect to \textit{uniform query}. However, such analyses do not consider the impact of the antenna array architecture on the LOS component of the channels, which herein we show to be significant. 
    Finally, the performance analyses conducted in \cite{Lopez.2019} were under the idealistic assumption of channels sharing the same mean phase.
\end{itemize} 
\subsection{Contributions and Organization of the Paper}\label{Cont}
This paper builds on CSI-free WET with multiple transmit antennas to power efficiently a large set of IoT devices. Different from our early work in 
\cite{Lopez.2019}, herein we do consider the mean phase shifts between antenna elements, which is a practical and unavoidable phenomenon. Based on such modeling we arrive to conclusions that are similar in some cases but different in others to those in \cite{Lopez.2019}. Specifically, the main contributions of this work can be listed as follows:
\begin{itemize}
	\item We analyze the CSI-free multi-antenna strategies studied in \cite{Lopez.2019}, e.g. $\mathrm{AA-SS}$ and $\mathrm{SA}$, in addition to the \textit{All Antennas transmitting Independent Signals} ($\mathrm{AA-IS}$) scheme, but under shifted mean phase  channels.  
	We do not consider any other information related to devices such as topological deployment, battery charge; although, such information could be crucial in
	some setups. Our derivations are specifically relevant for scenarios where it is difficult and/or not worth obtaining such information, e.g. when powering a massive number of low-power EH devices   with null/limited feedback to the PB; 
	\item By considering the non-linearity of the EH receiver we demonstrated that those devices far from the PB and more likely to operate near their sensitivity level, benefit more from the $\mathrm{SA}$ scheme than from $\mathrm{AA-IS}$. However, those closer to the PB and more likely to operate near saturation, benefit more from $\mathrm{AA-IS}$; 
	\item We attain the distribution and some main statistics of the RF energy at the EH receiver in correlated Rician fading channels under each WET scheme.  Notice that the Rician fading assumption is general enough to include a 	class of channels, ranging from Rayleigh fading channel 	without LOS to a fully deterministic LOS 	channel, by varying the Rician factor $\kappa$;	
	\item While $\mathrm{AA-IS}$ and $\mathrm{SA}$ cannot take advantage of the multiple antennas to improve the average statistics of the incident RF power,
	the energy dispersion can be significantly reduced, thus reducing the chances of energy outage. Meanwhile, the gains attained by $\mathrm{AA-SS}$ in terms of average RF energy delivery depend critically on the mean phase shifts between the antenna elements. In that regard, we show the considerable performance gaps between the idealistic $\mathrm{AA-SS}$ analyzed in \cite{Lopez.2019} and this scheme when considering channels with different mean phases. Even under such performance degradation, $\mathrm{AA-SS}$ still provides the greatest average harvested energy when compared to $\mathrm{AA-IS}$ and $\mathrm{SA}$ but its associated energy outage probability is generally the worst;
	\item We attained the optimum preventive phase shifting for maximizing the average energy delivery or minimizing its dispersion for each of the schemes. We found that when transmitting the same signal simultaneously over all the antennas ($\mathrm{AA-SS}$, or equivalently $\mathrm{AA}$ in \cite{Lopez.2019}), consecutive antennas must be $\pi$ phase-shifted for optimum performance. Meanwhile, under different schemes there is no need of carrying out any preventive phase shifting. Notice that all the analyzed CSI-free schemes, in their non-optimized form, say without proactive phase shifting, are special cases of the \textit{uniform} ($\mathrm{AA-SS}$) and \textit{unitary query} schemes ($\mathrm{AA-IS}$ and $\mathrm{SA}$) in the context of backscattering communications \cite{He.2015}, however, they  differ conceptually (and also in terms of performance in case of $\mathrm{AA-SS}$) when phase shifts are applied to the signals; 
	\item Our numerical results corroborate that correlation is beneficial under $\mathrm{AA-SS}$, especially under poor LOS where channels are more random. A very counter-intuitive result is that a greater LOS and/or number of antennas is not always beneficial when transmitting the same signal simultaneously through all antennas. 	 
	Meanwhile, since correlation (LOS and number of antennas) is well-known to decrease (increase) the diversity, $\mathrm{AA-IS}$ and $\mathrm{SA}$ schemes are affected by (benefited from) an increasing correlation (LOS factor and number of antennas);
	\item While most of the analytical derivations are obtained under the assumption that devices' positioning information is not available and/or they are uniformly distributed in the area, we show how the analyzed schemes can still be efficiently utilized to fairly power massive deployments when such assumptions do not hold. For instance, $\mathrm{AA-SS}$ ($\mathrm{SA}$ and $\mathrm{AA-IS}$) is (are) preferable when devices are (are not) clustered in specific spatial directions. 
\end{itemize}

Next, Section~\ref{system} presents the system model, while Section~\ref{csi-free}
presents and discusses the CSI-free WET strategies. Their performance under Rician fading is investigated in Sections~\ref{sAASS} and \ref{AA-IS}, while Section\ref{results} presents numerical results. Finally, Section~\ref{conclusions}
concludes the paper.
\newline
\textbf{Notation:} Boldface lowercase letters denote column vectors, while boldface uppercase letters denote matrices. For instance, $\mathbf{x}=\{x_i\}$, where $x_i$ is the $i$-th element of vector $\mathbf{x}$; while $\mathbf{X}=\{X_{i,j}\}$, where $X_{i,j}$ is the $i$-th row $j$-th column element of matrix $\mathbf{X}$. By $\mathbf{I}$ we denote the identity matrix, and by $\mathbf{1}$ we denote a vector of ones. Superscripts $(\cdot)^T$ and $(\cdot)^H$ denote the transpose and conjugate transpose operations, while $\mathrm{det}(\cdot)$ is the determinant, and by $\mathrm{diag}(x_1,x_2,\cdots,x_n)$ we denote the diagonal matrix with elements $x_1,x_2,\cdots x_n$. $\mathcal{C}$ and $\mathcal{R}$ are the set of complex and real numbers, respectively; while $\mathbbm{i}=\sqrt{-1}$ is the imaginary unit. Additionally, $|\cdot|$ and $\mathrm{mod}(a, b)$ are the absolute and modulo operations, respectively, while $||\mathbf{x}||$ denotes the euclidean norm of $\mathbf{x}$. $\mathbb{E}[\!\ \cdot\ \!]$ and $\mathrm{var}[\!\ \cdot\ \!]$ denote expectation and variance, respectively, while $\Pr[A]$ is the probability of event $A$. $\mathbf{v}\sim\mathcal{N}(\bm{\mu},\mathbf{R})$ and $\mathbf{w}\sim\mathcal{CN}(\bm{\mu},\mathbf{R})$ are a Gaussian real random vector and a circularly-symmetric Gaussian complex random vector, respectively, with mean vector $\bm{\mu}$ and covariance matrix $\mathbf{R}$. Additionally, $p_Y(y)$ denotes the PDF of random variable (RV) $Y$, while $Z\sim\chi^2(m,n)$ is a non-central chi-squared RV with $m$ degrees of freedom and parameter $n$. Then, according to \cite[Eq.(2-1-125)]{Proakis.2001} the first two central moments are given by
$\mathbb{E}[Z]=m+n$, and
$\mathrm{var}[Z]=2(m+2n)$. 
Finally, $J_0(\cdot)$ denotes the Bessel function of first kind and order $0$ \cite[\S 10.2]{Thompson.2011}.
\section{System model}\label{system}
Consider the scenario in which a PB equipped with $M$ antennas powers wirelessly a large set $\{S_i\}$ of single-antenna sensor nodes located nearby. 
Since this work deals only with CSI-free WET schemes, and for such scenarios the characterization of one sensor is representative of the overall performance, we focus our attention to the case of a generic node $S$. The fading channel coefficient between the $j-$th PB's antenna and $S$ is denoted as $h_{j}\in\mathcal{C},\ j\in\{
1,\cdots,M\}$, while $\mathbf{h}\in\mathcal{C}^{M\times 1}$ is a vector with all the antennas' channel coefficients.
\subsection{Channel model}\label{sysA}
Quasi-static channels are assumed, where the fading process is considered to be constant over the transmission of a block  and independent and identically distributed (i.i.d) from block to block.
Without loss of generality we set the duration of a block to 1 so the terms energy and power can be indistinctly used.
Specifically, we consider channels undergoing Rician fading, which is a very general assumption that allows  modeling a wide variety of channels by tuning the Rician factor $\kappa\ge 0$ \cite[Ch.2]{Proakis.2001}, e.g. when $\kappa=0$ the channel envelope is Rayleigh distributed, while when $\kappa\rightarrow\infty$ there is a fully deterministic LOS channel. Therefore,
\begin{align}
\mathbf{h}=\sqrt{\frac{\kappa}{1+\kappa}}e^{\mathbbm{i}\varphi_0}\mathbf{h}_{\mathrm{los}}+\sqrt{\frac{1}{1+\kappa}}\mathbf{h}_{\mathrm{nlos}}\label{channel}
\end{align}
is the normalized channel vector \cite[Ch.5]{Hampton.2013}, where 
$\mathbf{h}_\mathrm{los}=[1, e^{\mathbbm{i}\Phi_1},\cdots,e^{\mathbbm{i}\Phi_{M-1}}]^T$
is the deterministic LOS propagation component such that $\Phi_t,\ t\in\{1,\cdots,M-1\}$  is the mean phase shift of the $(t+1)-$th array element with respect to the first antenna. 
 Additionally, $\varphi_0$ accounts for an initial phase shift, while $\mathbf{h}_\mathrm{nlos}\sim \mathcal{CN}(\mathbf{0},\mathbf{R})$ represents the scattering (Rayleigh) component.
We assume a real covariance matrix $\mathbf{R}\in\mathcal{R}^{M\times M}$ for gaining in analytical tractability, which means that real and imaginary parts of $\mathbf{h}_\mathrm{nlos}$ are i.i.d  and also with covariance $\mathbf{R}$ \cite{Schreier.2010}. 
Assume half-wavelength equally-spaced antenna elements, e.g. as in a uniform linear array (ULA), yielding 
\begin{align}\label{Phi}
\Phi_t&=-t\pi\sin\phi,
\end{align}
where $\phi\in[0,2\pi]$ is  the azimuth angle relative to the boresight of the transmitting antenna array. 
Such angle depends on both transmit and receive local conditions, e.g. antenna orientation, node's location, and consequently it is different for each sensor $S$. 
Additionally, let use denote by $\beta$ the average RF power available at $S$ if the PB transmits with full power over a single antenna. Under such single-antenna setup, the available RF energy at the input of the EH circuitry is given by $\beta|h_{i^*}|^2$, where $i^*$ is the index of the active antenna. 
Notice that $\beta$ includes the effect of both path loss and transmit power.  
\subsection{Preventive adjustment of mean phases}
As mentioned earlier, $\phi$ and consequently $\bm{\Phi}$, are different for each sensor, which prevents us from making any preventive phase adjustment based on an specific $\phi$. However, maybe we could still use the  topological information embedded in \eqref{Phi}, which tell us that $\Phi_t$ increases with $t$ for a given $\phi$, to improve the statistics of the harvested energy. To explore this, let us consider that the PB applies a preventive adjustment of the signal phase at the $(t+1)-$th array element given by $\psi_t\in[0,2\pi]$, while without loss of generality we set $\psi_0=0$. Then, the equivalent normalized channel vector seen at certain sensor $S$ becomes $\mathbf{h}^*=\Psi \mathbf{h}$, where  
	\begin{align}
	\Psi&=\mathrm{diag}(1,e^{\mathbbm{i}\psi_1},\cdots,e^{\mathbbm{i}\psi_{M-1}}).
	\end{align}
	Now, departing from \eqref{channel} we have that
	\begin{align}\label{channel2}
	\mathbf{h}^*&=\sqrt{\frac{\kappa}{1+\kappa}}e^{\mathbbm{i}\varphi_0}\Psi\mathbf{h}_\mathrm{los}+\sqrt{\frac{1}{1+\kappa}}\Psi\mathbf{h}_\mathrm{nlos}\nonumber\\
	&\sim \sqrt{\frac{\kappa}{1+\kappa}}e^{\mathbbm{i}\varphi_0} \Big[1, e^{\mathbbm{i}(\Phi_1+\psi_1)},\cdots,e^{\mathbbm{i}(\Phi_{M-1}+\psi_{M-1})}\Big]^T+\nonumber\\
	&\qquad\qquad\qquad\qquad\qquad\qquad + \sqrt{\frac{1}{1+\kappa}}\mathcal{CN}(\mathbf{0,\mathbf{R}}),
	\end{align}
	where last line comes from simple algebraic operations and using the fact that $\Psi\mathbf{h}_\mathrm{nlos}\sim \mathcal{CN}(\mathbf{0},\Psi\mathbf{R}\Psi^H)\sim \mathcal{CN}(\mathbf{0},\mathbf{R})$ since $\Psi$ is diagonal with unit absolute values' entries.
	Without loss of generality, by conveniently setting $\varphi_0=\pi/4$ \cite{Khalighi.2001} so $e^{\mathbbm{i}\varphi_0}=(1+\mathbbm{i})/\sqrt{2}$ imposes the effect of LOS (constant) component on real and imaginary parts of the scattering (Rayleigh) component $\mathbf{h}_\mathrm{nlos}$, we rewrite \eqref{channel2} as $\mathbf{h}^*=\mathbf{h}_x+\mathbbm{i}\mathbf{h}_y$, where $\mathbf{h}_x$ and $\mathbf{h}_y$ are independently distributed as
		\begin{align}
		\mathbf{h}_{x,y}&\sim \sqrt{\frac{1}{2(\kappa+1)}} \mathcal{N}\big(\sqrt{\kappa}\bm{\omega}_{x,y} ,\mathbf{R}\big),\label{hxy}
		\end{align}
	where 
	$\bm{\omega}_{x,y}\!=\!\big[1,\cos(\Phi_1\!+\!\psi_1)\mp\sin(\Phi_1\!+\!\psi_1),\cdots,\cos(\Phi_{M-1}+\psi_{M-1})\mp\sin(\Phi_{M-1}+\psi_{M-1})\big]^T$.
\subsection{EH transfer function}\label{EH}
Finally, and after going through the channel, the RF energy is harvested at the receiver end. The EH circuitry is characterized by a non-decreasing function $g: \mathcal{R}^+\mapsto \mathcal{R}^+$ modeling the relation between the incident and harvested RF power\footnote{In practice $g$ also depends on the modulation and incoming waveform \cite{Clerckx.2016,Clerckx.2018}.} at $S$. 
With such a function in place, the power transfer efficiency (PTE) is given by $g(x)/x$.
In most of the works $g$ is assumed to be linear for analytical tractability, e.g. \cite{Ju.2014,Lopez.2017,LopezAlves.2019,ChenLi.2015,Lopez.2018,LopezDemo.2017,Bacinoglu.2018,Li.2018,Huang.2014,Huang.2016,Park.2017,Zeng.2015,ZengZhang.2015,BiZhang.2016}, which implies that $g(x)=\eta x$ where $\eta$ is a PTE constant, thus, independent of the input power. However, in practice the PTE actually depends on the input power and consequently, the relationship between the input power and the output power is nonlinear \cite{Clerckx.2019,Boshkovska.2015,Chen.2017,Lopez.2019}. 
In this work, we consider the following EH transfer function \cite{Boshkovska.2015}
\begin{align}
g(x)&=g_{\max} \Big(\frac{1+e^{ab}}{1+e^{-a(x-b)}}-1\Big)e^{-ab},\label{gx}
\end{align}
which is known to describe accurately the non-linearity of EH circuits by properly fitting parameters $a,b\in\mathcal{R}^+$, while $g_{\max}$ is the harvested power at saturation. 
\section{CSI-free multi-antenna WET strategies}\label{csi-free}
Herein we overview CSI-free multiple-antenna WET strategies for an efficient wireless powering, while discussing some related practicalities.
\subsection{All Antennas transmitting the Same Signal (AA-SS)}\label{S-AASS}
Under this scheme, the PB transmits the same signal $s$ simultaneously with all antennas and with equal power at each. Such scheme (in its non-optimized form) is referred to as \textit{uniform query} in the context of backscattering communications \cite{He.2015}, and as $\mathrm{AA}$  in \cite{Lopez.2019}. Herein we refer to it as $\mathrm{AA-SS}$ to explicitly highlight  its difference with respect to the $\mathrm{AA-IS}$ scheme (see next subsection). Under this scheme, the RF signal at the receiver side and ignoring the noise, whose energy is negligible for harvesting\footnote{In WET setups, the performance depends on the available energy at the input of the energy harvester, which requires to be significant, at least in the order of sub-$\mu$W in case of high sensitive EH hardware. Therefore, the noise impact is practically null, as widely recognized in the literature, e.g., \cite{Hou.2018,Clerckx.2019,Ju.2014,Lopez.2017,LopezAlves.2019,ChenLi.2015,Lopez.2018,LopezDemo.2017,LopezE.2018,Huang.2016,Park.2017,Zeng.2015,ZengZhang.2015,BiZhang.2016,Lopez.2019,Boshkovska.2015}.}, is given by $\sum_{j=1}^M \sqrt{\frac{\beta}{M}}h_j^* s$, where $s$ is normalized such that $\mathbb{E}[s^Hs]=1$. 
Then, the energy harvested by $S$ is given by 
\begin{align}
\label{aass}
\xi_{\mathrm{aa-ss}}&=g(\xi_{\mathrm{aa-ss}}^\mathrm{rf}), \, \mathrm{where}\\
\label{aassf}
\xi_{\mathrm{aa-ss}}^\mathrm{rf}&=\mathbb{E}_s\bigg[\Big(\sum_{j=1}^M \sqrt{\frac{\beta}{M}}h_j^* s\Big)^H\Big(\sum_{j=1}^M \sqrt{\frac{\beta}{M}}h_j^* s\Big)\bigg] \nonumber \\ &=\mathbb{E}_s\bigg[\Big|\sum_{j=1}^M \sqrt{\frac{\beta}{M}}h_j^*\Big|^2 s^Hs\bigg]\nonumber\\
&=\Big|\sum_{j=1}^M \sqrt{\frac{\beta}{M}}h_j^*\Big|^2 \mathbb{E}[s^Hs]=\frac{\beta}{M}\big|\mathbf{1}^T\mathbf{h}^*\big|^2
\end{align}
is the available RF energy. Notice that the terms energy and power can be used indistinctly since the block duration is normalized and the PB does not change its strategy over time.
\subsection{All Antennas transmitting Independent Signals (AA-IS)} \label{S-AAIS}
Instead of transmitting the same signal over all antennas, the PB may transmit signals $s_j$ independently generated across the antennas. This is for alleviating the issue of destructive signal combination at $s$, and constitutes a special case of the so-called \textit{unitary query} in the context of backscattering communications \cite{He.2015}. 
We refer to this scheme as AA-IS, for which the RF signal at the receiver side and ignoring the noise, whose energy is negligible for harvesting, is given by $\sum_{j=1}^{M}\sqrt{\frac{\beta}{M}}h_j^*s_j$, where each $s_j$ is normalized such that $\mathbb{E}[s_j^Hs_j]=1, \ \forall j$. Then, the harvested energy is given by 
\begin{align}
\label{aais} 
\xi_\mathrm{aa-is}&=g(\xi_\mathrm{aa-is}^\mathrm{rf}), \, \mathrm{where}\\
\label{aaisf}
\xi_\mathrm{aa-is}^\mathrm{rf}&=\mathbb{E}_s\bigg[\Big(\sum_{j=1}^M \sqrt{\frac{\beta}{M}}h_j^* s_j\Big)^H\Big(\sum_{j=1}^M \sqrt{\frac{\beta}{M}}h_j^* s_j\Big)\bigg]\nonumber\\
&=\mathbb{E}_s\bigg[\sum_{j=1}^M \Big|\sqrt{\frac{\beta}{M}}h_j^*\Big|^2 s_j^Hs_j\bigg]=\frac{\beta}{M}\sum_{j=1}^M \big|h_j^*\big|^2 \mathbb{E}[s_j^Hs_j]\nonumber\\
&=\frac{\beta}{M}||\mathbf{h}^*||^2.
\end{align}

For both, $\mathrm{AA-SS}$ and $\mathrm{AA-IS}$, the signal power over each antenna is $1/M$ of the total available transmit power. Additionally, notice that $M$ RF chains are required since all the antennas are simultaneously active. This is different from the CSI-free scheme discussed next.
\subsection{Switching Antennas (SA)}\label{SA}
Instead of transmitting with all antennas at once, the PB may transmit the (same or different) signal $s$ with full power by one
antenna at a time such that all antennas are used during a block. This is the SA scheme analyzed in \cite{Lopez.2019}, while it also constitutes a special case of the \textit{unitary query} method in backscattering communications \cite{He.2015}. In this case just one RF chain is required, hence, reducing circuit power consumption, hardware complexity and consequently the economic cost. 

Assuming equal-time allocation for each antenna, the system is equivalent to that in which each sub-block duration is $1/M$ of the total block duration, and the total harvested
energy accounts for the sum of the $M$ sub-blocks. The RF signal at the receiver side during the $j-$th sub-block, and ignoring the noise, whose energy is negligible for harvesting, is given by $\sqrt{\beta}h_j^*s$, where $s$ is normalized such that $\mathbb{E}[s^Hs]=1$, then
\begin{align}\label{sa}
\xi_\mathrm{sa}&=\frac{1}{M}\sum_{j=1}^Mg(\xi_{\mathrm{sa},j}^\mathrm{rf}), \, \mathrm{where}\\
\xi_{\mathrm{sa},j}^\mathrm{rf}&=\mathbb{E}_s\Big[\Big(\sqrt{\beta}h_j^*s\Big)^H\Big(\sqrt{\beta}h_j^*s\Big)\Big]=\mathbb{E}_s\Big[\Big|\sqrt{\beta}h_j^*\Big|^2s^Hs\Big]\nonumber\\
&=\beta |h_j^*|^2\mathbb{E}_s\big[s^Hs\big]=\beta |h_j^*|^2
\end{align}
is the incident RF power during the $j-$th sub-block.

Notice that for the simple, but commonly adopted in literature, linear EH model, both \eqref{aais} and \eqref{sa} match. 
However, in practice $g(x)$ is non-linear, and consequently \eqref{aais} may differ significantly from \eqref{sa}.
We depart from \eqref{gx} to write the second derivative of $g(x)$ as
\begin{align}
\frac{d^2}{d x^2}g(x)&=\frac{a^2e^{ax}(1+e^{ab})(e^{ab}-e^{ax})g_{\max}}{(e^{ab}+e^{ax})^3},
\end{align}
which allows us to conclude that $g$ is convex (concave) for $x\le b$ ($x\ge b$). Then, for certain channel vector realization $\mathbf{h}$ and using Jensen's inequality, we have that
\begin{align}
g\Big(\frac{\beta}{M}\sum\limits_{j=1}^{M}|h_j|^2\Big)&\left\{\begin{array}{ll}
\!\!\le \frac{1}{M}\sum\limits_{j=1}^{M}g\big(\beta|h_j|^2\big), & \mathrm{if}\ |h_j|^2\le \frac{b}{\beta}\ \forall j\\
\!\!\ge \frac{1}{M}\sum\limits_{j=1}^{M}g\big(\beta|h_j|^2\big), & \mathrm{if}\ |h_j|^2\ge \frac{b}{\beta}\ \forall j
\end{array}\right.\!\!,\nonumber\\
\xi_\mathrm{aa-is}&\left\{\begin{array}{ll}
\!\!\le \xi_\mathrm{sa}, & \mathrm{if}\ \max|h_j|^2\le \frac{b}{\beta}\\
\!\!\ge \xi_\mathrm{sa}, & \mathrm{if}\ \min|h_j|^2\ge \frac{b}{\beta}
\end{array}\right.\!\!. \label{SAeq}
\end{align}
\begin{remark}\label{re1}
	Above result implies that devices far from the PB and more likely to operate near their sensitivity level, benefit more from the $\mathrm{SA}$ scheme than from $\mathrm{AA-IS}$. However, those closer to the PB and more likely to operate near saturation, benefit more from $\mathrm{AA-IS}$.
\end{remark}

In the following we analyze the statistics of the RF energy available at $S$ for harvesting under the $\mathrm{AA-SS}$ and $\mathrm{AA-IS}$ schemes. For such schemes, the harvested energy comes from mapping the RF energy through the EH transfer function $g$, as shown in \eqref{aass} and \eqref{aais}. For $\mathrm{SA}$ the mapping is much more convoluted, especially because the available RF energy varies (although possibly correlationally) within the same coherence block as illustrated in \eqref{sa}. However, our analysis in the previous paragraphs suggests that the statistics of $\xi_\mathrm{aa-is}$ and $\xi_\mathrm{sa}$ may be approximated, which is an issue we discuss numerically in detail in Section~\ref{results}.
\section{RF available energy under AA-SS}\label{sAASS}
Next, we characterize the distribution of the RF power at $S$ under the $\mathrm{AA-SS}$ scheme.
\begin{theorem}\label{the1}
	Conditioned on the mean phase shifts of the powering signals, the distribution of the RF power at the input of the energy harvester under the $\mathrm{AA-SS}$ operation is given by
	\begin{align}
	\xi^{\mathrm{rf}}_{\mathrm{aa-ss}}\sim \frac{\beta R_{\scriptscriptstyle\sum}}{2(\kappa+1)M}\chi^2\Big(2,\frac{2\kappa f(\bm{\psi},\phi)}{R_{\scriptscriptstyle\sum}}\Big), \label{AA-SS}
	\end{align}
	where $R_{\scriptscriptstyle\sum}=\bm{1}^T\mathbf{R}\ \!\bm{1}$,
	\begin{align}
	f(\bm{\psi},\phi)&=\upsilon_1(\bm{\psi},\phi)^2+\upsilon_2(\bm{\psi},\phi)^2,\label{fv}
	\end{align}
	\begin{equation}
	\begin{array}{r@{}l}	
	\upsilon_1(\bm{\psi},\phi)&=1+\sum_{t=1}^{M-1}\cos\big(\psi_t+\Phi_t\big)\\
		\upsilon_2(\bm{\psi},\phi)&=\sum_{t=1}^{M-1}\sin\big(\psi_t+\Phi_t\big)
	\end{array},\label{v1v2}
	\end{equation}
	and $\Phi_t$ is given in \eqref{Phi} as a function of $\phi$.
\end{theorem}
\begin{proof}
	See Appendix~\ref{App-A}. \phantom\qedhere
\end{proof}

Notice that \eqref{AA-SS} matches \cite[Eq.(25)]{Lopez.2019} just in the specific case of un-shifted mean phases, e.g. $\bm{\psi}+\mathbf{\Phi}=\mathbf{0}$. In such case $f(\bm{\psi},\phi)$ becomes $M^2$.
 \subsection{On the impact of different phase means}
The impact of different phase means on the system performance is strictly determined by $f(\bm{\psi},\phi)$ in \eqref{AA-SS}, and can be better understood by checking the main statistics, e.g. mean and variance, of the incident RF power, which can be easily obtained from \eqref{AA-SS} by using as
\begin{align}
\mathbb{E}\big[\xi_\mathrm{aa-ss}^\mathrm{rf}\big]&=\frac{\beta R_{\scriptscriptstyle\sum}}{2(\kappa+1)M}\Big(2+\frac{2\kappa}{ R_{\scriptscriptstyle\sum}}f(\bm{\psi},\phi)\Big)\nonumber\\
&=\frac{\beta}{M(\kappa+1)}\big( R_{\scriptscriptstyle\sum}+\kappa f(\bm{\psi},\phi)\big),\label{av_aass}\\
\mathrm{var}\big[\xi_\mathrm{aa-ss}^\mathrm{rf}\big]&=\frac{\beta^2 R_{\scriptscriptstyle\sum}^2}{4(\kappa+1)^2M^2}\Big(4+\frac{8\kappa}{ R_{\scriptscriptstyle\sum}}f(\bm{\psi},\phi)\Big)\nonumber\\
&=\frac{\beta^2 R_{\scriptscriptstyle\sum}}{(\kappa+1)^2M^2}\Big( R_{\scriptscriptstyle\sum}+2\kappa f(\bm{\psi},\phi)\Big).\label{va_aass}
\end{align}
 \begin{figure}[t!]
	\centering  \includegraphics[width=0.9\columnwidth]{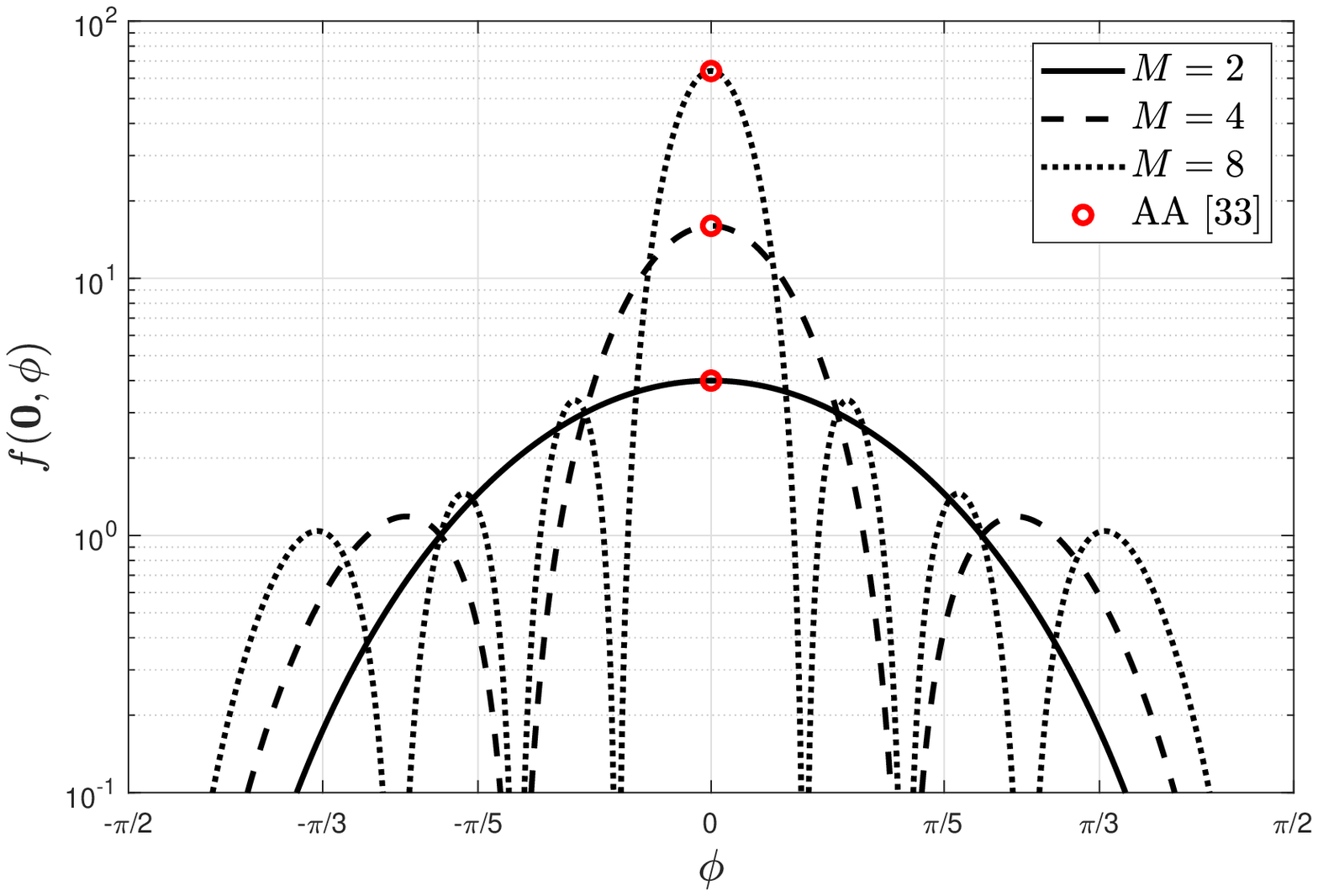}\\
	\includegraphics[width=0.9\columnwidth]{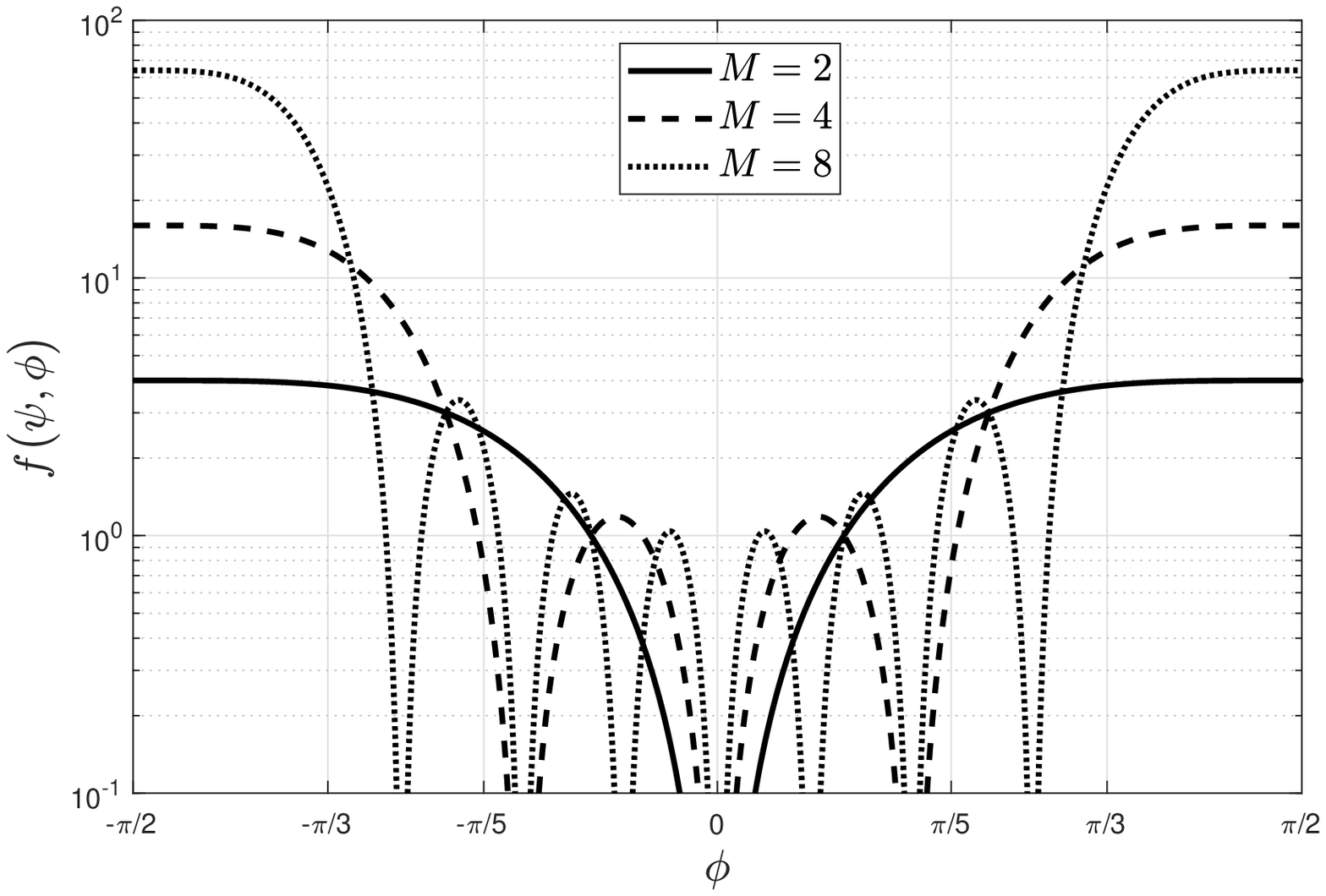}
	\caption{$a)$ $f(\bm{0},\phi)$ vs $\phi$ (top), and $b)$ $f(\bm{\psi},\phi)$ vs $\phi$ for $\bm{\psi}$ according to \eqref{psit} (bottom). We set $M\in\{2,4,8\}$.}		
	\label{Fig23}
\end{figure}
Therefore, both mean and variances increases with $f(\bm{\psi},\phi)$. Meanwhile, it is easy to check that $f(\bm{\psi},\phi)$  is maximized for $\bm{\psi}+\mathbf{\Phi}=\mathbf{0}$, for which $f(\bm{\psi},\phi)=M^2$, thus, the entire analysis carried out in \cite{Lopez.2019} on this AA-SS scheme provides upper-bounds for both the mean and variance of the harvested energy. However, different phase means cause in practice a degradation on the diversity order of $\xi^{\mathrm{rf}}$. 
Let us assume no preventive adjustment of mean phases is carried out, e.g. $\bm{\psi}=\bm{0}$, to illustrate in Fig.~\ref{Fig23}a the impact of such different channel phase means. Specifically, we show $f(\bm{0},\phi)$ for different values of $M$. Notice that the performance diverges fast from the one claimed in \cite{Lopez.2019} as $\phi$ moves away from $0$. 
\begin{remark}\label{rr2}
	The number of minima of $f(\bm{\psi},\phi)$ matches $M$, thus, as $M$ increases the chances of operating close to a minimum increase as well, which deteriorates significantly the system performance in terms of average incident RF power.
\end{remark}
\subsection{Preventive adjustment of mean phases}\label{preventive1}
Herein, we discuss on how to set the vector $\bm{\psi}$ for optimizing the system performance. 
For clarity, and using \eqref{v1v2} followed by some algebraic transformations, we rewrite \eqref{fv} as follows
\begin{align}
f(\bm{\psi},\phi)
&=\!\bigg(\!1\!+\!\sum_{t=1}^{M-1}\!\cos\big(\Phi_t\!+\!\psi_t\big)\!\bigg)^2\!\!\!+\!\bigg(\!\sum_{t=1}^{M-1}\!\sin\big(\Phi_t\!+\!\psi_t\big)\!\bigg)^2\nonumber\\
&=M+2\sum_{t=1}^{M-1}\cos(\Phi_t+\psi_t)+\nonumber\\
&\qquad +2\sum_{t=1}^{M-2}\sum_{l=t+1}^{M-1}\!\cos\big(\Phi_t+\psi_t-\Phi_l-\psi_l\big).\label{f}
\end{align}
Assume $\phi$ uniformly distributed in $[0,2\pi]$\footnote{This fits scenarios where the $\theta$ corresponding to each sensor is unknown, or alternatively, scenarios where there is a very large number of sensors homogeneously distributed in space such that  $p_{\phi}(\phi)\approx\frac{1}{2\pi}$. Although our analysis here holds specifically for such uniform angle distribution, our procedures and ideas can be extended to other scenarios.}, e.g. $p_{\phi}(\phi)=\frac{1}{2\pi}$, then the problem translates to optimize over $f(\bm{\psi})=\frac{1}{2\pi}\int_{0}^{2\pi}f(\bm{\psi},\phi)\mathrm{d}\phi$. Substituting  \eqref{Phi} into \eqref{f} and integrating over $\phi$ we attain
\begin{align}
f(\bm{\psi})&=M+2\sum_{t=1}^{M-1}J_0(t\pi)\cos \psi_t+\nonumber\\
&\qquad+2\sum_{t=1}^{M-2}\sum_{l=t+1}^{M-1}J_0((t-l)\pi)\cos(\psi_t-\psi_l),\label{f2}
\end{align}
which comes from using the integral representation of $J_0(\cdot)$ \cite[Eq.(10.9.1)]{Thompson.2011}. Obviously, \eqref{av_aass} and \eqref{va_aass} still hold but using $f(\bm{\psi})$ instead of $f(\bm{\psi},\phi)$.
\subsubsection{Average energy maximization}
Solving the problem of maximizing the average incident RF power is equivalent to solve $\argmax{\bm{\psi}}f(\bm{\psi})$. Now, since $|\cos \alpha|\le 1$  we have that
\begin{align}
f(\bm{\psi})\!\le\!M\!+\!2\sum_{t=1}^{M\!-\!1}\!\big|J_0(t\pi)\big|\!+\!2\sum_{t=1}^{M\!-\!2}\sum_{l=t\!+\!1}^{M-1}\!\big|J_0((t\!-\!l)\pi)\big|.\label{fu2}
\end{align} 
Using the fact that $J_0(t\pi)$ is positive (negative) if $t$ is even (odd), we can easily observe that the upper bound in \eqref{fu2} can be attained by setting $\psi_t=0$ ($\pi$) for $t$ even (odd) in \eqref{f2}, e.g.
\begin{align}\label{psit}
\psi_t=\mod\!(t,2)\pi,\qquad t\in\{0,1,\cdots,M-1\}.
\end{align}
\begin{remark}\label{re3}
	Above results means that consecutive antennas must be $\pi$ phase-shifted for optimum average energy performance under the $\mathrm{AA-SS}$ scheme.
\end{remark}
Despite optimum $f(\bm{\psi})$ as in \eqref{fu2} cannot be simplified, an accurate approximation is 
\begin{align}
f(\bm{\psi})& 
\approx 0.85\times M^{1.5}, \label{f22}
\end{align}
which comes from standard curve-fitting. 

In Fig.~\ref{Fig23}b we show the impact of above preventive phase shifting on $f(\bm{\psi},\phi)$ for different angles $\phi$. By comparing Fig.~\ref{Fig23}a (no preventive phase shifting) and Fig.~\ref{Fig23}b (preventive phase shifting given in \eqref{psit}), notice that $i)$ the number of minima keeps the same, $ii)$ the best performance occurs now for $\phi=\pm \pi/2$, while the worst situation happens when $\phi=0$; and $iii)$ there are considerable improvements in terms of area under the curves, which are expected to conduce to considerable improvements when averaging over $\phi$.
\subsubsection{Energy dispersion minimization}\label{Ed}
As metric of dispersion we consider the variance. Therefore, herein we aim to solve $\argmin{\bm{\psi}}f(\bm{\psi})$; although notice that this in turns minimize the average available RF energy as well\footnote{It is easily verifiable that the phase shifting that satisfies $\argmin{\bm{\psi}}f(\bm{\psi})$, minimizes the  \textit{coefficient of variation} parameter, which is given by $\frac{\sqrt{\mathrm{var}[\xi]}}{\mathbb{E}[\xi]}$, and constitutes probably a more suitable dispersion metric.}.

Different from the maximization problem, the problem of minimizing $f(\bm{\psi})$ is not such easy to handle. We resorted to Matlab numerical solvers and realized that optimal solutions diverge significantly for different values of $M$. However, we found that $f(\bm{0})$ approximates extraordinarily to the optimum function value. Therefore, $\bm{\psi}=\bm{0}$ \textit{minimizes} the variance of the incident RF power, and not preventive phase shifting is required in this case. This means that
\begin{align}
f(\bm{\psi})&\gtrsim f(\mathbf{0})\nonumber\\
&=M+2\sum_{t=1}^{M-1}J_0(t\pi)+2\sum_{t=1}^{M-2}\sum_{l=t+1}^{M-1}J_0((t-l)\pi)\nonumber\\
&\stackrel{(a)}{\approx }0.64 \times M,\label{Map}
\end{align}
where $(a)$ comes from standard curve-fitting.

\begin{remark}\label{re4}
 Results in \eqref{f22} and \eqref{Map} evidence that both $\max f(\bm{\psi})$ and $\min f(\bm{\psi})$ share a polynomial dependence on $M$. In case of $\min f(\bm{\psi})$ such relation is linear, while $\max f(\bm{\psi})$ is roughly $\sqrt{M}$ times greater than $\min f(\bm{\psi})$.
\end{remark}
\begin{figure}[t!]
	\centering  \includegraphics[width=0.9\columnwidth]{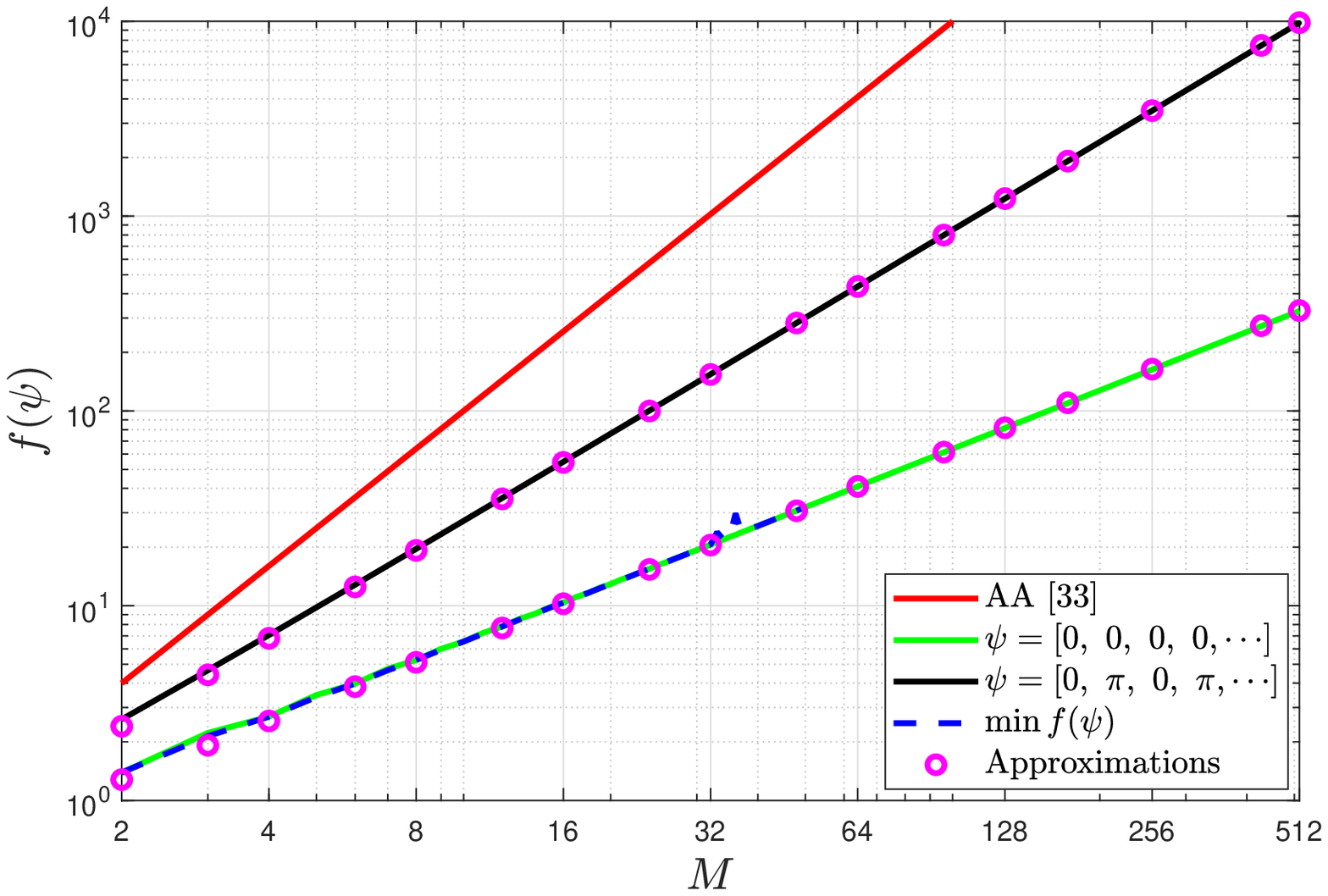}
	\caption{$f(\bm{\psi})$ vs $M$.}		
	\label{Fig_app}
\end{figure}
\subsubsection{Validation}
In Fig.~\ref{Fig_app} we show $f(\bm{\psi})$ as a function of $M$ for i) the phase shifting in \eqref{psit}, which maximizes the average incident RF power, ii) $\bm{\psi}=\mathbf{0}$, e.g. no preventive phase shifting, which minimizes both the variance and average statistics, and iii) the scenario discussed in \cite{Lopez.2019} which is constrained to $\bm{\psi}=\mathbf{\Phi}=\mathbf{0}$. We utilized Matlab numerical optimization solvers for minimizing $f(\bm{\psi})$, but such approach was efficient just for $M\le 32$. For greater $M$, Matlab solvers do not always converge and are extremely time-consuming. Meanwhile, notice that $\bm{\psi}=\mathbf{0}$ indeed approaches extremely to $\arg\min_{\bm{\psi}} f(\bm{\psi})$. 
Approximations in \eqref{f22} and \eqref{Map} are accurate. 
\subsubsection{On the optimization gains}
By using \eqref{av_aass}, \eqref{f22} and \eqref{Map}, the dB gain in average incident RF energy with respect to the non preventive shifting scheme, which in turns minimizes the energy dispersion according to our discussion in Subsection~\ref{Ed}, can be obtained as follows
\begin{align}\label{d1}
\delta_{\mathbb{E}}&\approx 10\log_{10}\bigg(\frac{\beta\big(R_{\scriptscriptstyle\sum}+0.85\kappa M^{1.5}\big)}{M(\kappa+1)}\bigg/\frac{\beta\big(R_{\scriptscriptstyle\sum}+0.64\kappa M\big)}{M(\kappa+1)}\bigg)\nonumber\\
&=10\log_{10}\Big(\frac{ R_{\scriptscriptstyle\sum}+0.85\kappa M^{1.5}}{ R_{\scriptscriptstyle\sum}+0.64\kappa M}\Big)\nonumber\\
&\ge 10\log_{10}\Big(\frac{M+0.85\kappa \sqrt{M}}{M+0.64\kappa }\Big),
\end{align}
where last line comes from the fact that the argument of the logarithm is a decreasing function of $ R_{\scriptscriptstyle\sum}$ since $0.85\kappa M^{1.5}>0.64\kappa M$ for $M\ge 2$, and $ R_{\scriptscriptstyle\sum}\le M^2$. 

Now, we evaluate the costs in terms of the variance increase of such average energy maximization shifting. By using \eqref{va_aass}, \eqref{f22} and \eqref{Map} we have that
\begin{align}\label{d12}
\delta_{\mathrm{var}}&\approx 10\log_{10}\Big(\frac{ R_{\scriptscriptstyle\sum}+2\times 0.85\kappa M^{1.5}}{ R_{\scriptscriptstyle\sum}+2\times 0.64\kappa M}\Big)\nonumber\\
&\stackrel{(a)}{\le} 10\log_{10}\Big(\frac{2\times 0.85\kappa \sqrt{M}}{2\times 0.64\kappa }\Big)\nonumber\\
&=5\log_{10}M+1.23,
\end{align}
where $(a)$ comes from using the lower-bound of $ R_{\scriptscriptstyle\sum}$, e.g. $ R_{\scriptscriptstyle\sum}\ge 0$. 
Above results imply for instance that the gain in the average incident RF energy is above 3.47 dB when $\kappa=10$ and $M=8$, while the variance can increase $5.75$ dB maximum as well.
\section{RF available energy under AA-IS}\label{AA-IS}
Next, we characterize the distribution of the RF power at the receiver end under the $\mathrm{AA-IS}$ scheme.
\begin{theorem}\label{the2}
	Conditioned on the mean phase shifts of the powering signals, the approximated distribution of the RF power available as input to the energy harvester under the $\mathrm{AA-IS}$ is 
	\begin{align}\label{aa-is2}
	\!\xi_\mathrm{aa-is}^\mathrm{rf}&\!\sim\!\!\frac{\beta}{2M^2(\kappa\!+\!1)}\! \Bigg(R_{\scriptscriptstyle\sum}\chi^2\bigg(2,\frac{2\kappa f(\bm{\psi},\phi)}{ R_{\scriptscriptstyle\sum}}\bigg)+\nonumber\\
	 +&\frac{M^2\!\!-\! R_{\scriptscriptstyle\sum}}{M-1}\chi^2\bigg(\!2(M\!-\!1),\!\frac{2M\!(M\!-\!1)\kappa \tilde{\upsilon}(\bm{\psi},\phi)}{M^2- R_{\scriptscriptstyle\sum}}\bigg)\! \Bigg), 
	\end{align}
	where
	\begin{align}
	\tilde{\upsilon}&(\bm{\psi},\phi)=\! M\!-\!1\!+\!2\sum_{j=1}^{M-1}\frac{1}{j(j\!+\!1)}\bigg(\sum_{t\!=\!M\!-\!j\!+\!1}^{M-1}\!\!\!\cos\big(\psi_t\!+\!\Phi_t\big)\!+\nonumber\\
	& \!\!+\!\!\!\sum_{t\!=\!M\!-\!j\!+\!1}^{M-1}\sum_{l=\!t\!+\!1}^{M-1}\!\!\cos\!\big(\psi_t\!+\!\Phi_t\!-\!\psi_l\!-\!\Phi_l\big)\!\!-\!j\cos\!\big(\psi_{M\!-\!j}\!+\!\Phi_{M\!-\!j}\big)\!+\nonumber\\
	&\qquad\qquad -j\!\!\sum_{t=M-j+1}^{M-1}\!\!\!\cos(\psi_{M\!-\!j}\!+\!\Phi_{M\!-\!j}\!-\!\psi_t\!-\!\Phi_t)\bigg).\label{up1}
	\end{align}
\end{theorem}
\begin{proof}
	See Appendix~\ref{App-B}. \phantom\qedhere
\end{proof}
For $\bm{\psi}+\mathbf{\Phi}=\mathbf{0}$, \eqref{aa-is2} matches \cite[Eq.(39)]{Lopez.2019}\footnote{Notice that \cite[Eq.(39)]{Lopez.2019} describes the distribution of the RF incident power  under un-shifted mean phases and the SA operation considering the whole block time. Therefore, both \eqref{aa-is2} and \cite[Eq.(39)]{Lopez.2019} must match when $\mathbf{\Phi}=\mathbf{0}$ according to our discussions in Subsection~\ref{SA}.}. However, 
when mean phase shifts between antenna elements increase, both expressions diverge.
\subsection{On the impact of different mean phases}
It is not completely clear from \eqref{aa-is2} whether phase shifts are advantageous or not under this scheme.
Let us start by checking the average statistics of $\xi_\mathrm{aa-is}^\mathrm{rf}$ as follows
\begin{align}\label{av_aais}
\mathbb{E}\big[\xi_\mathrm{aa-is}^\mathrm{rf}\big]&\approx \frac{\beta}{2M^2(\kappa+1)}\Bigg(\frac{M^2- R_{\scriptscriptstyle\sum}}{M-1}\bigg(2(M-1)+\nonumber\\
&+\!\frac{2M(M\!-\!1)\kappa \tilde{\upsilon}(\bm{\psi},\phi)}{M^2- R_{\scriptscriptstyle\sum}}\bigg)\!+\! R_{\scriptscriptstyle\sum}\bigg(2+\frac{2\kappa f(\bm{\psi},\phi)}{ R_{\scriptscriptstyle\sum}}\bigg)\Bigg)\nonumber\\
&=\frac{\beta}{M^2(\kappa+1)}\Big(M^2-R_{\scriptscriptstyle\sum}+M\kappa\tilde{\upsilon}(\bm{\psi},\phi)+\nonumber\\
&\qquad\qquad\qquad\qquad\qquad\qquad +R_{\scriptscriptstyle\sum}+\kappa f(\bm{\psi},\phi)\Big)\nonumber\\
&=\beta\Big(1+\frac{\kappa \tilde{f}(\bm{\psi},\phi)}{M^2(\kappa+1)}\Big),
\end{align}
where $\tilde{f}(\bm{\psi},\phi)=f(\bm{\psi},\phi)+M\tilde{\upsilon}(\bm{\psi},\phi)-M^2$. Notice that the larger $\tilde{f}(\bm{\psi},\phi)$, the greater $\mathbb{E}\big[\xi_\mathrm{aa-is}^\mathrm{rf}\big]$. However, $\tilde{f}(\bm{\psi},\phi)\lll 1$ independently of the value of $\phi$ as shown in Fig.~\ref{Fig34}a for the case where no preventive adjustment of mean phases is carried out, e.g. $\bm{\psi}=\bm{0}$. 
\begin{remark}
	Therefore, channel mean phase shifts do not \textit{strictly} bias the average harvested energy, which is intuitively expected since transmitted signals are independent to each other. Meanwhile, such result is very different from what happened under the $\mathrm{AA-SS}$ scheme for which mean phase shifts always influenced (negatively) on such metric.
\end{remark}
Additionally, notice that when $ R_{\scriptscriptstyle\sum}$ is maximum (perfect correlation), $\mathrm{AA-SS}$ could provide up to $M$ times more energy on average than $\mathrm{AA-IS}$, whose average statistics are not affected in any way by $ R_{\scriptscriptstyle\sum}$. 
Previous statement holds as long as there is not a strong LOS component; however, as $\kappa$ takes greater values, which is typical of WET systems due to the short range characteristics, $f(\bm{\psi},\phi)$ and $\tilde{f}(\bm{\psi},\phi)$ become more relevant, which favors the $\mathrm{AA-IS}$ scheme.

Let us investigate now the impact on the variance as follows
\begin{align}
\mathrm{var}[\xi_\mathrm{aa-is}^\mathrm{rf}]&\approx \frac{\beta^2}{2M^4(\kappa+1)^2}\Bigg(\frac{(M^2- R_{\scriptscriptstyle\sum})^2}{(M-1)^2}\bigg(2(M-1)+\nonumber\\
&+\!\frac{4M\!(M\!-\!1)\kappa \tilde{\upsilon}(\bm{\psi},\phi)}{M^2- R_{\scriptscriptstyle\sum}}\bigg)\!\!+\! R_{\scriptscriptstyle\sum}^2\bigg(\!2\!+\!\frac{4\kappa f(\bm{\psi},\phi)}{ R_{\scriptscriptstyle\sum}}\bigg)\!\Bigg)\nonumber\nonumber\\
&\approx \frac{\beta^2}{M^3(M\!-\!1)(\kappa\!+\!1)^2}\Big(M^3(1\!+\!2\kappa)\!+\! R_{\scriptscriptstyle\sum}^2+\nonumber\\
&-\!2M R_{\scriptscriptstyle\sum}(1\!+\!\kappa)\!+\!2\kappa M( R_{\scriptscriptstyle\sum}\!-\!M)f(\bm{\psi},\phi)\Big),\label{va2}
\end{align}
where last line comes from taking advantage of $f(\bm{\psi},\phi)+M\tilde{\upsilon}(\bm{\psi},\phi)-M^2\approx 0\rightarrow \tilde{\upsilon}(\bm{\psi},\phi)\approx M-f(\bm{\psi},\phi)/M$ to write $\mathrm{var}[\xi_\mathrm{aa-is}^\mathrm{rf}]$ just as a function of $f(\bm{\psi},\phi)$. Since $f(\bm{\psi},\phi)\ge 0$, it is obvious that just when $M> R_{\scriptscriptstyle\sum}$, e.g. negative correlation of some antennas, the system performance benefits from having different mean phases.
\subsection{Preventive adjustment of mean phases}\label{preventive2}
Herein, we discuss on how to set the vector $\bm{\psi}$ for optimizing the system performance. Again, $\phi$ is taken randomly and uniformly from $[0,2\pi]$ as in Subsection~\ref{preventive1}.
\subsubsection{Average energy maximization}
As shown in Fig.~\ref{Fig34}a, channel mean phase shifts do not impact significantly on the average incident RF power. Then, it is intuitively expected that none preventive phase shifting would improve such average statistics. To corroborate this we require computing $\tilde{f}(\bm{\psi})=\frac{1}{2\pi}\int_{0}^{2\pi}\tilde{f}(\bm{\psi},\phi)\mathrm{d}\phi$, for which
\begin{align}
\tilde{f}(\bm{\psi})&=\frac{1}{2\pi}\int\limits_{0}^{2\pi}f(\bm{\psi},\phi)\mathrm{d}\phi \!+\! \frac{M}{2\pi}\int\limits_{0}^{2\pi}\tilde{\upsilon}(\bm{\psi},\phi)\mathrm{d}\phi-\frac{M^2}{2\pi}\int\limits_{0}^{2\pi}\mathrm{d}\phi\nonumber\\
&=f(\bm{\psi})\!-\!M\!+\!\!\sum_{j=1}^{M-1}\!\frac{4\pi}{j(j+1)}\bigg(\sum_{t\!=\!M\!-\!j\!+\!1}^{M-1}\!\!J_0(t\pi)\cos\psi_t+\nonumber\\
&\qquad+\sum_{t=M-j+1}^{M-1}\sum_{l=t+1}^{M-1}J_0\big((l-t)\pi\big)\cos(\psi_t-\psi_l)+\nonumber\\
&\qquad-j\!\!\sum_{t=M-j+1}^{M-1}\!\!J_0\big((t-M+j)\pi\big)\cos(\psi_{M-j}-\psi_t)+\nonumber\\
&\qquad\qquad\qquad\qquad\!-jJ_0\big((M\!-\!j)\pi\big)\cos\psi_{M\!-\!j}\bigg),\label{fpsi}
\end{align}
where $f(\bm{\psi})$ is given in \eqref{f2} and last line comes from integrating $\tilde{\upsilon}(\bm{\psi},\phi)$ over $\phi$ by using the integral representation of $J_0(\cdot)$ \cite[Eq.(10.9.1)]{Thompson.2011}. Due to the extreme non-linearity of \eqref{fpsi} we resort to exhaustive search optimization solvers of Matlab to find $\arg\max_{\bm{\psi}}\tilde{f}(\bm{\psi})$ and compare its performance with the non preventive phase shifting scheme for which $\bm{\psi}=\bm{0}$. Specifically, we show in Fig.~\ref{Fig34}b their associated performance in terms of $\tilde{f}({\psi})$ normalized by $M^2$ since the average incident RF power depends strictly on such ratio. Results evidence that although the performance gap is large in relative terms, it is not in absolute values. That is, not even the optimum preventive  phase shifting allows increasing $\tilde{f}({\psi})/M^2$ significantly, e.g. $\tilde{f}({\psi})/M^2\approx 0$.
\begin{remark}\label{re6}
	Then, and based on \eqref{av_aais}, the average incident RF power is approximately $\beta$. Therefore, this scheme cannot take advantage of the multiple antennas to improve the average statistics of the incident RF power in any way.
\end{remark}
\begin{figure}[t!]
	\centering  \includegraphics[width=0.9\columnwidth]{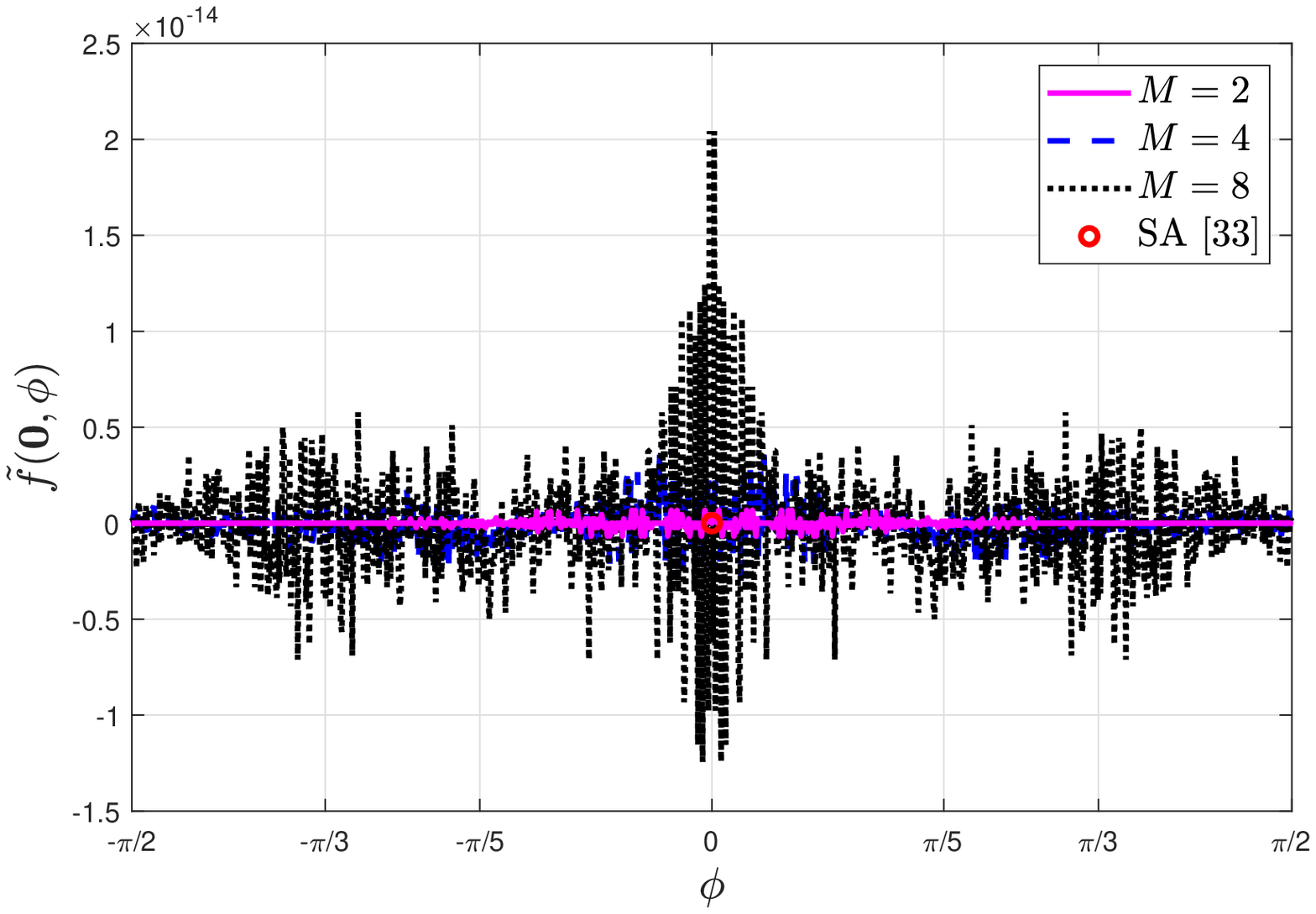}\\	
	\centering  \includegraphics[width=0.9\columnwidth]{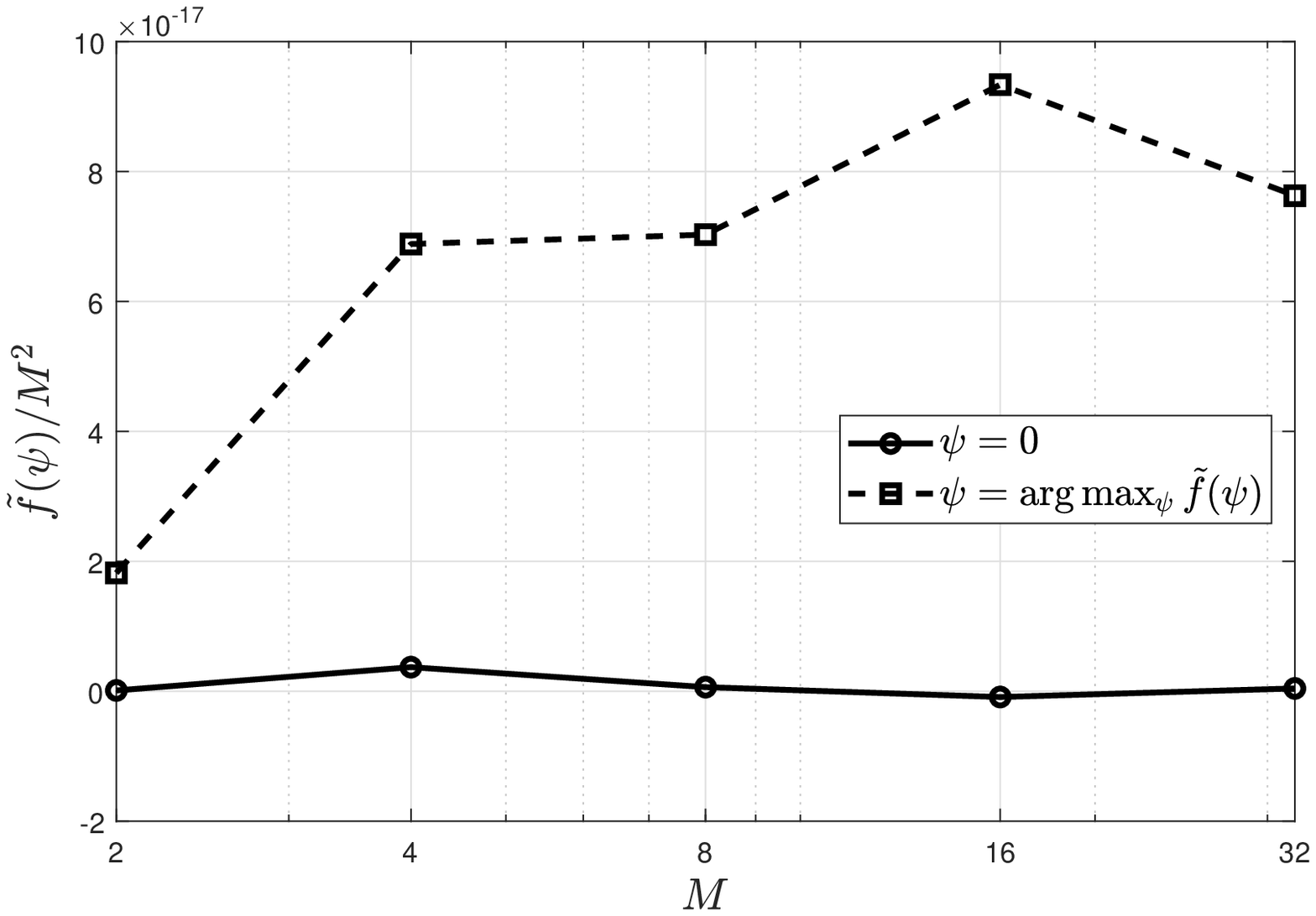}		
	\caption{$a)$ $\tilde{f}(\bm{0},\phi)$ vs $\phi$ for $M\in\{2,4,8\}$ (top), and $b)$ $\tilde{f}(\bm{\psi})$ vs $M$ for $\bm{\psi}=\mathbf{0}$ and $\bm{\psi}=\arg\max_{\bm{\psi}}\tilde{f}(\bm{\psi})$ (bottom).}
	\label{Fig34}
\end{figure} 
\subsubsection{Energy dispersion minimization}
As in Subsection~\ref{Ed} we consider the variance of the incident RF power as the dispersion measure. Then, we have that
\begin{align}
\arg\min_{\bm{\psi}}&\ \mathrm{var}[\xi_{\mathrm{aa-is}}^\mathrm{rf}]\nonumber\\
&=\left\{\begin{array}{ll}
\arg\min_{\bm{\psi}} f(\bm{\psi}), & \mathrm{if}\ R_{\scriptscriptstyle\sum}>M\\
\arg\max_{\bm{\psi}} f(\bm{\psi}), & \mathrm{if}\ R_{\scriptscriptstyle\sum}<M
\end{array}\right.\nonumber\\
&=\left\{\begin{array}{ll}
\bm{\psi}\approx\mathbf{0}, & \mathrm{if}\ R_{\scriptscriptstyle\sum}>M\\
\psi_t=\mathrm{mod}(t,2)\pi, & \mathrm{if}\ R_{\scriptscriptstyle\sum}<M
\end{array}\right.,\label{opt}
\end{align}
where last line comes from using directly our previous results in Subsection~\ref{preventive1}. In addition, notice that  $\mathrm{var}[\xi_{\mathrm{aa-is}}^\mathrm{rf}]$ is not a function of $f(\bm{\psi})$ when $R_{\scriptscriptstyle\sum}=M$ (see \eqref{va2}), thus, a preventive  phase shifting is not necessary as it does not make any difference. 
\begin{remark}
	Since the average incident RF energy is not affected by any phase shifting, we can conclude that \eqref{opt} provides the optimum preventive phase shifting. This means that phase shifting is not required when $R_{\scriptscriptstyle\sum}\ge M$, as in most of the practical systems. 
\end{remark}
On the other hand, observe from \eqref{va2} that the variance decreases with $M$, thus, although the $\mathrm{AA-IS}$ is not more advantageous than single antenna transmissions in terms of the provided average RF energy, it benefits significantly from the multiple antennas to reduce the energy dispersion.

Finally, based on Subsection~\ref{SA} it is expected that the phase shifting given in \eqref{opt} approaches the optimum for the SA scheme as well, hence we adopt it in such scenario.
\section{Numerical Results}\label{results}
Herein we present simulation results on the performance of the discussed CSI-free multiple-antenna schemes under the non-linear EH model given in \eqref{gx}.
We evaluate the average harvested energy\footnote{Notice that the average PTE is just a scaled version of such average harvested energy for a fixed transmit power, and therefore its corresponding curves follow the same trends. Similarly the PDF and CDF of the harvested energy can be easily extrapolated to the statistics of the instantaneous PTE.}, and energy outage probability, which refers to the probability that the RF energy falls below an energy threshold $\xi_0$, thus, interrupting the devices' operation\footnote{$\xi_0$ must be obviously not smaller than the EH sensitivity.}. Notice the energy outage probability is highly related with both mean and variance statistics.
The EH hardware parameters 
agree with the EH circuitry experimental data at $2.45$GHz in \cite{Khan.2019}.

 \begin{table}[!t]
	\centering
	\caption{CSI-free schemes and corresponding preventive phase shifting}
	\begin{tabular}{lc}
		\toprule
		\textbf{Schemes} & \textbf{$\psi_t$}  \\
		\midrule
		$\mathrm{AA-SS}_{\max \mathrm{E}}$  & $\mathrm{mod}(t,2)\pi$  \\ 
		$\mathrm{AA-SS}_{\min \mathrm{var}}$  & $0$  \\
		$\mathrm{AA-IS}$  & $0$  \\
		$\mathrm{SA}$  & $0$  \\
		''scheme'' ($\phi=0$) \cite{Lopez.2019} (unrealistic) & $0$ ($\phi=0$)  \\
		\bottomrule
	\end{tabular}\label{table}
\end{table}

We take $\phi$ uniformly and randomly from $[0,2\pi]$, while comparing the corresponding results to those with equal mean phases, $\phi=0$, which are not attainable in practice and are just presented as benchmark. 
For clarity we summarize  the schemes under consideration in Table~\ref{table}. 
Notice that in case of $\mathrm{AA-IS}$ and $\mathrm{SA}$ we just consider the preventive phase shifting given in the first line of \eqref{opt} since we assume positive spatial correlation, e.g. $ R_{\scriptscriptstyle\sum}\ge M$. Specifically, we assume  exponential spatial correlation with coefficient $\tau$ such
that $R_{i,j}=\tau^{|i-j|}$, hence, 
\begin{align}
 R_{\scriptscriptstyle\sum}\!=\!M\!+\!2\sum_{i=1}^{M-1}(M-i)\tau^i\!=\!\frac{M(1\!-\!\tau^2)\!-\!2\tau(1\!-\!\tau^M)}{(1-\tau)^2},
\end{align}
where last step comes from using a geometric series compact representation. 
Such model is physically reasonable since
correlation decreases with increasing distance between antennas. Unless stated otherwise we set the system parameters as shown in Table~\ref{table2}. Notice that $\kappa=5$ and $\tau=0.3$ to account for certain LOS and correlation, while we assume the PB is equipped with a moderate-to-small  number of antennas $M=8$. 
\subsection{On the distribution of the harvested energy}\label{ehs}
In Fig.~\ref{Fig_1r} we illustrate the PDF of the energy harvested under each of the schemes. Specifically, Fig.~\ref{Fig_1r}a shows the PDF for two different path-loss profiles, while considering $\mathrm{AA-IS}$ and $\mathrm{SA}$ schemes. We observe that under $\mathrm{SA}$ the harvested energy under large path loss presents slightly better statistics than under $\mathrm{AA-IS}$, while $\mathrm{AA-IS}$ is superior when operating under better average channel statistics. Therefore, we corroborate our statements in Remark~\ref{re1}: devices far from the PB indeed benefit more from $\mathrm{SA}$ than from $\mathrm{AA-IS}$, while those closer to the PB benefit more from the latter. We also show that the statistics improve by considering the channel mean phases, therefore, the performance of these schemes is better in a practical setup than the foreseen by \cite{Lopez.2019}. Something different occurs under $\mathrm{AA-SS}$ as shown in Fig.~\ref{Fig_1r}b. As we discussed in Section~\ref{sAASS}, the un-shifted mean phase assumption is the most optimistic under the operation of $\mathrm{AA-SS}$, and notice that the performance gains with respect to what can be attained in practice, e.g. under the $\mathrm{AA-SS}_\mathrm{min\ var}$ and $\mathrm{AA-SS}_\mathrm{max\ E}$ discussed in this work, are extremely notorious. Regarding  $\mathrm{AA-SS}_\mathrm{max\ E}$ vs $\mathrm{AA-SS}_\mathrm{min\ var}$, we can observe the performance gains in terms of average and variance of the harvested energy, respectively, of one with respect to the other. We must say that although $\mathrm{AA-SS}_\mathrm{min\ var}$ provides less disperse harvested energy values, they could be extremely small compared to the ones achievable under $\mathrm{AA-SS}_\mathrm{max\ E}$.
\begin{table}[!t]
	\centering
	\caption{Simulation parameters}
	\begin{tabular}{cc}
		\toprule
		\textbf{Parameter} & \textbf{Value}  \\
		\midrule
		$g_\mathrm{max}$  & $2$ mW  \\ 
		$(a,\ b)$  & $(0.56,\ 3.5)$  \\
		$\xi_0$  & $-2$ dBm  \\
		$\phi$ & uniformly random in $[0,2\pi]$\\
		$\kappa$  & $5$   \\
		$\tau$ & $0.3$  \\
		$M$ & $8$ \\
		\bottomrule
	\end{tabular}\label{table2}
\end{table}
\begin{figure}[t!]
	\centering  
	\includegraphics[width=1.05\columnwidth]{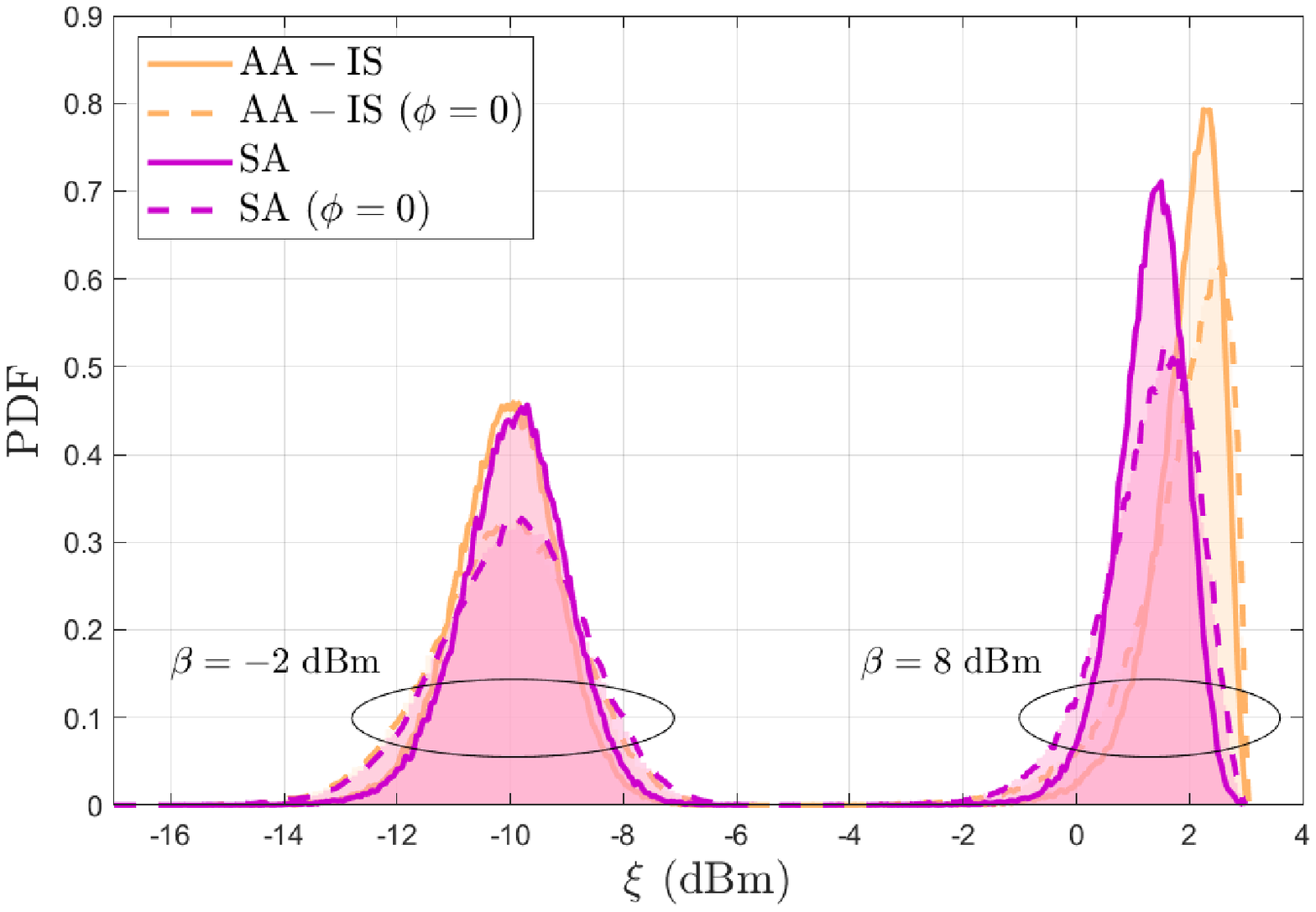}\\ 
	\includegraphics[width=1.05\columnwidth]{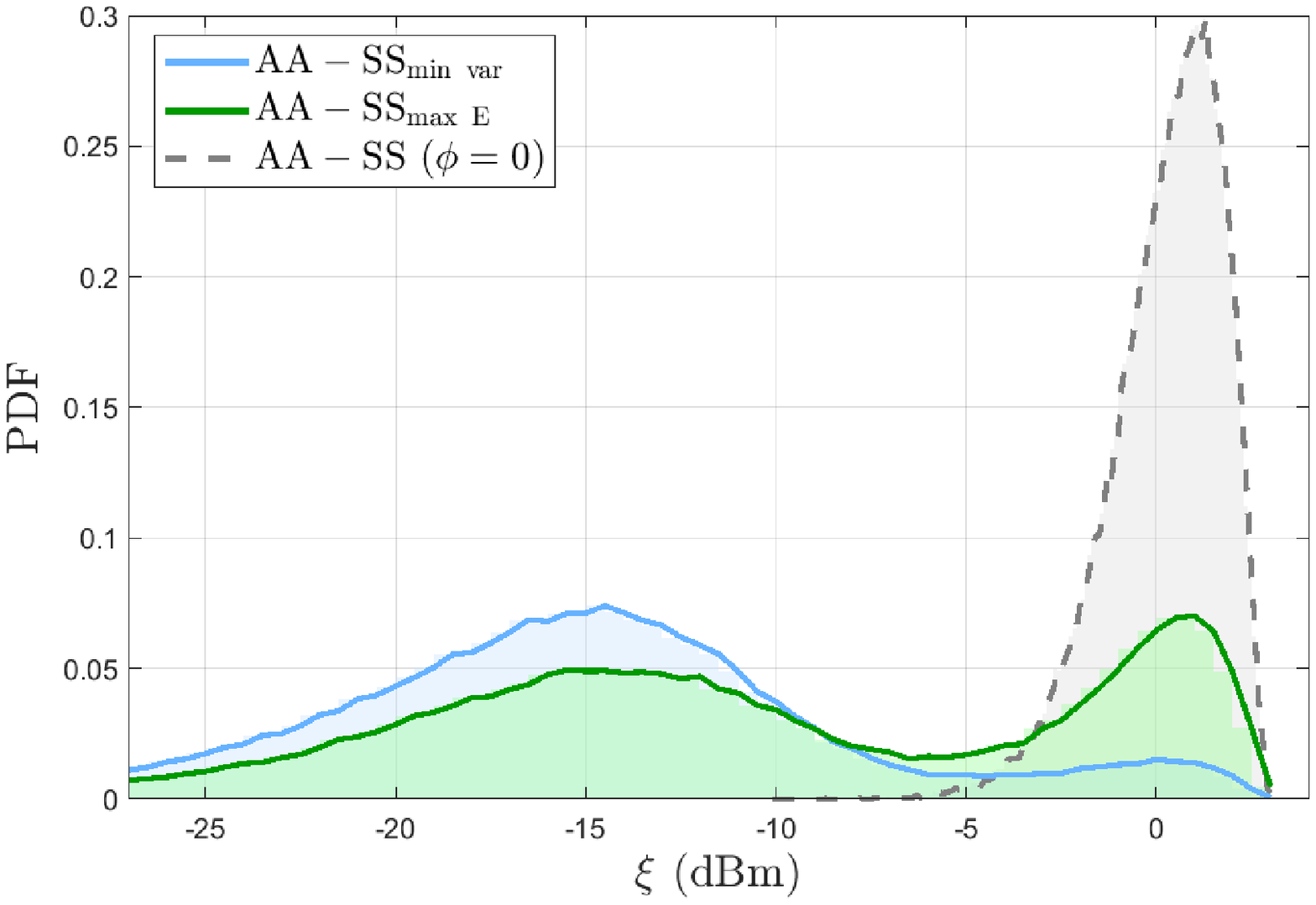}
	\caption{PDF of the harvested energy  $a)$ under $\mathrm{AA-IS}$ and $\mathrm{SA}$ for $\beta\in\{-2,8\}$ dBm (top), and $b)$ under $\mathrm{AA-SS}$ for $\beta=-2$ dBm (bottom).} 	
	\label{Fig_1r}
\end{figure} 

Let us analyze now the CDF curves, which are shown in  Fig.~\ref{Fig_2r} for $\beta=2$ dBm\footnote{In this and most of the subsequent figures we adopt $\beta=2$ dBm to illustrate the performance in the desired operation range, e.g. above/below sensitivity/saturation, since for an incident RF power of $2$ dBm ($1.6$ mW) the harvested power becomes $-5.5$ dBm ($0.28$ mW).}. 
Notice that, although with greater variance, $\mathrm{AA-SS}_\mathrm{max\ E}$ performs superior\footnote{In this context a superior performance means that the corresponding CDF curve is below.} than $\mathrm{AA-SS}_\mathrm{min\ var}$. In fact, this holds in most scenarios, which highlights the need of a $\pi-$phase shifting in consecutive antennas when using the $\mathrm{AA-SS}$ scheme as foreseen in Remark~\ref{re3}. Meanwhile, the greatest diversity order is attained by  $\mathrm{AA-IS}$ and $\mathrm{SA}$ schemes, which perform similarly. In this case, $\mathrm{AA-IS}$ performs slightly superior but this depends greatly on the path loss according to our discussions in previous paragraph. Notice that these two schemes guarantee, for instance, chances above $99\%$ of harvesting more than $-8$ dBm of power, compared to just $60\%$ and $40\%$ when using $\mathrm{AA-SS}_\mathrm{max\ E}$ and $\mathrm{AA-SS}_\mathrm{min\ var}$, respectively. Once again, the gaps with respect to the performance under un-shifted mean phases are evidenced, and are shown to be enormous especially when transmitting the same signal simultaneously by all antennas. 
\begin{figure}[t!]
	\centering  
	\includegraphics[width=0.9\columnwidth]{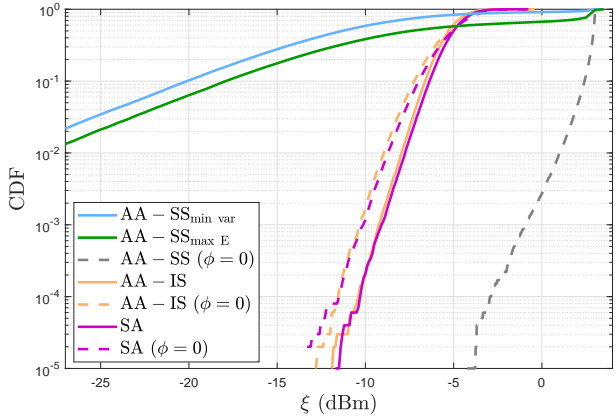}
	\vspace{-1mm}
	\caption{CDF of the harvested energy for $\beta=2$ dBm.}	 
	\vspace{-1mm}
	\label{Fig_2r}
\end{figure} 
\subsection{On the impact of the channel}
\begin{figure}[t!]
	\centering  
	\includegraphics[width=0.9\columnwidth]{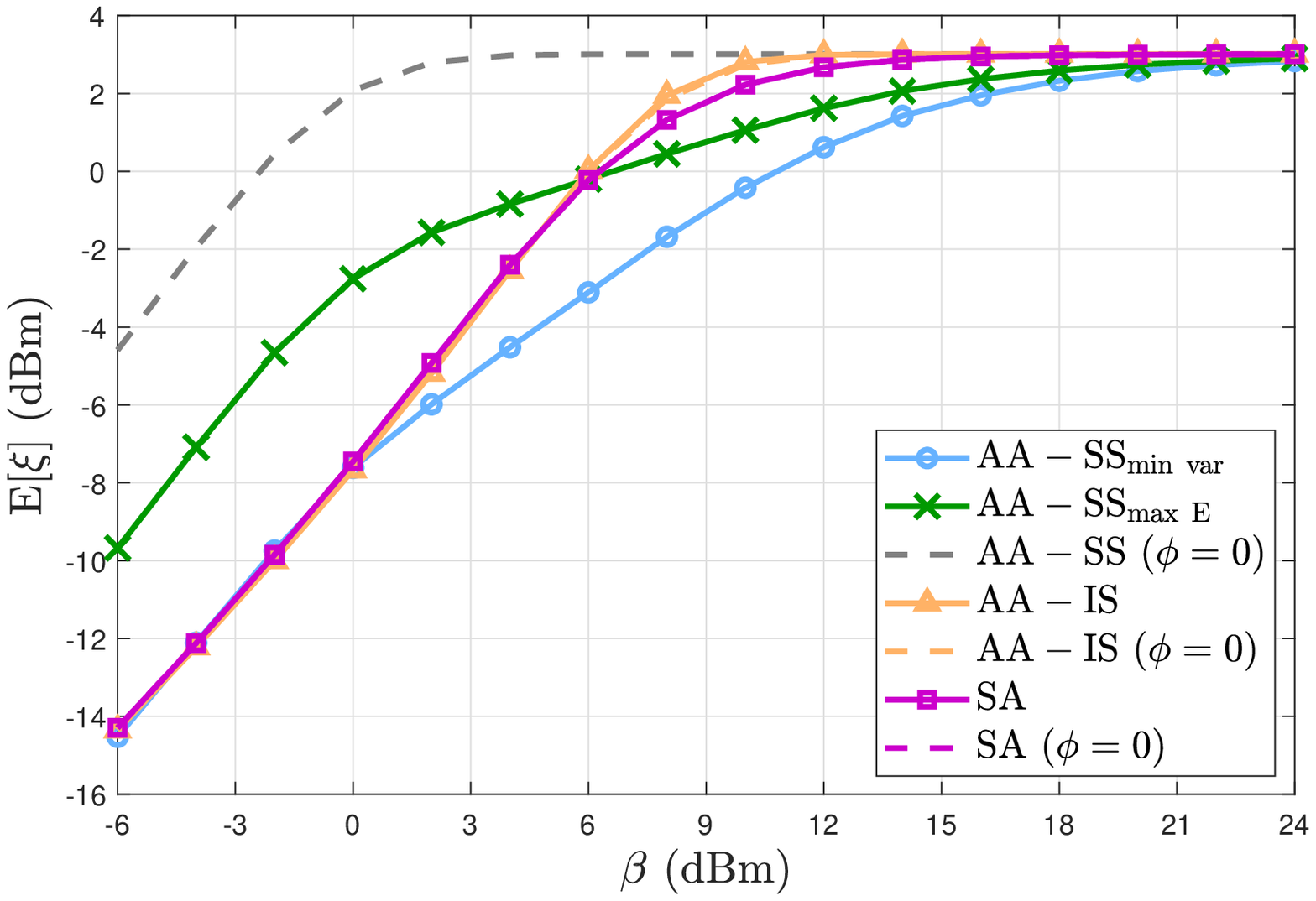}\\
	\includegraphics[width=0.9\columnwidth]{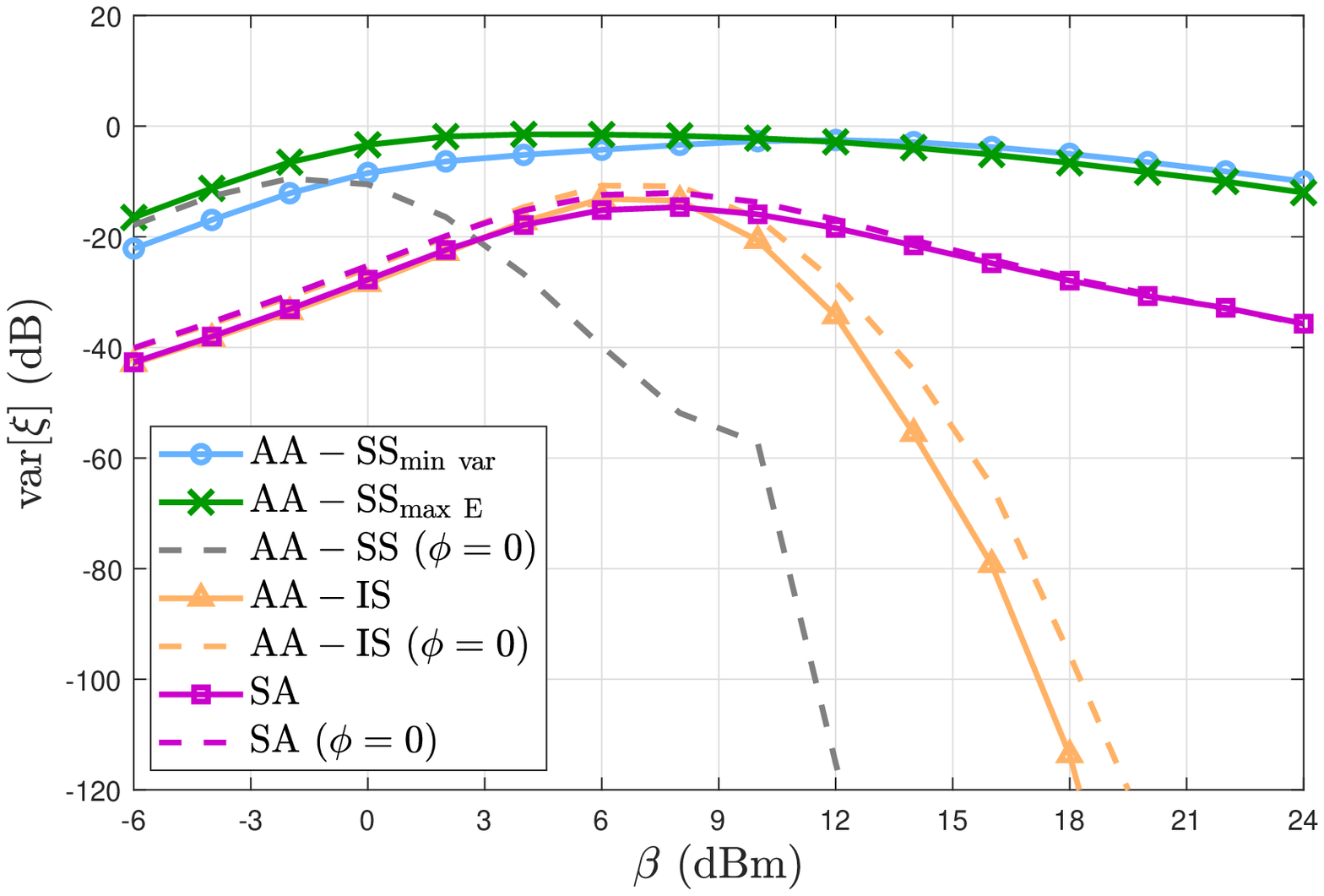}\\	
	\includegraphics[width=0.9\columnwidth]{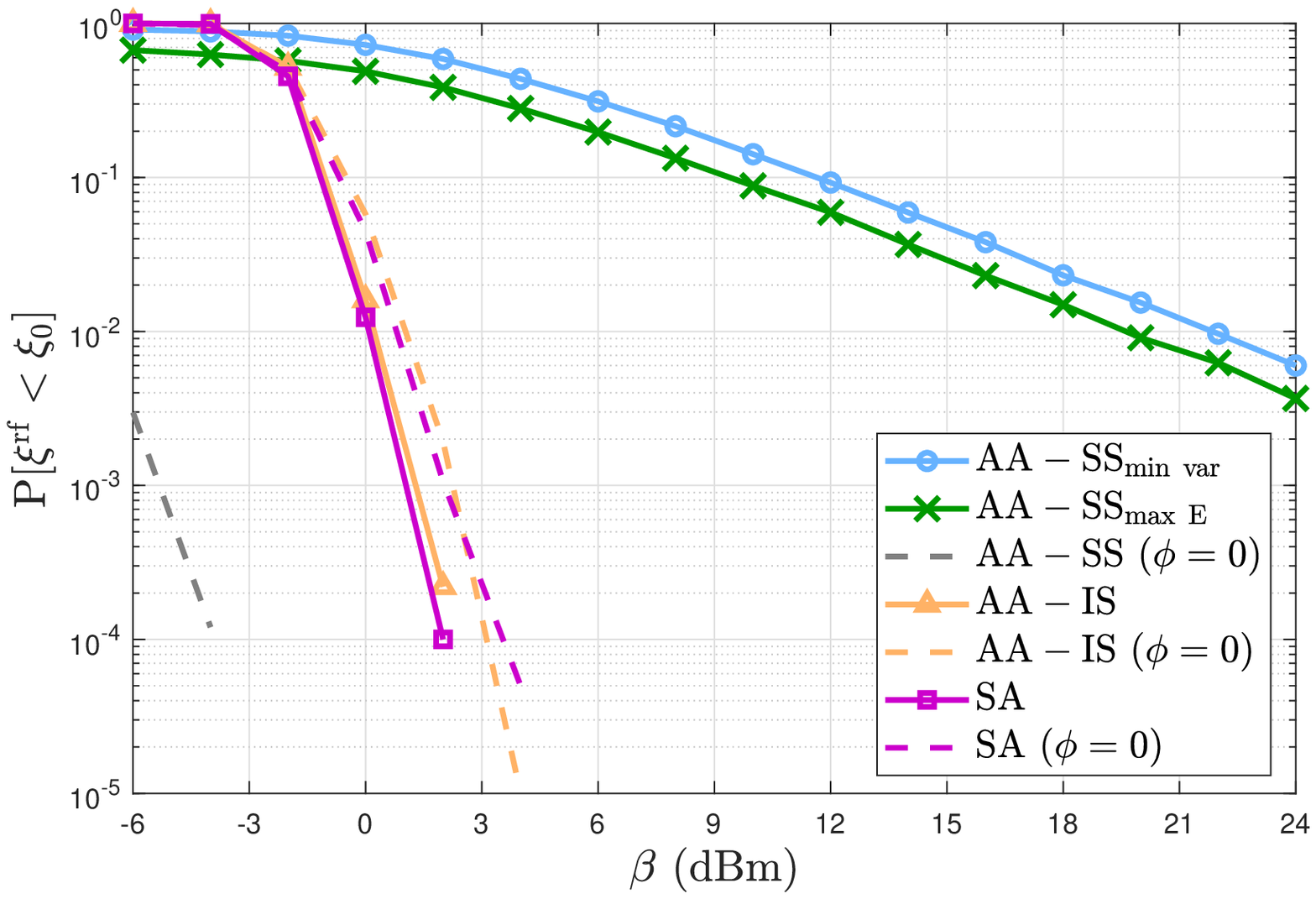}
	\caption{$a)$ Average harvested energy (top), $b)$ variance of the harvested energy (middle), and $c)$ energy outage probability (bottom), as a function of $\beta$.}		
	\label{Fig6}
\end{figure}
Fig.~\ref{Fig6} shows the average (Fig.~\ref{Fig6}a) and variance (Fig.~\ref{Fig6}b) of the harvested energy, and the energy outage probability (Fig.~\ref{Fig6}c), as a function of the path loss. In general, both the average and variance of the incident RF power increases with $\beta$ for all schemes. However, the saturation phenomenon captured by \eqref{gx} makes that the average harvested energy saturates at some point, while obviously the variance would decrease from that point onward. These phenomena are observed in Fig.~\ref{Fig6}a and  Fig.~\ref{Fig6}b. Notice that $\mathrm{AA-SS}_\mathrm{max\ E}$ achieves considerable gains in terms of average harvested energy compared to $\mathrm{AA-SS}_\mathrm{min\ var}$\footnote{Based on Remark~\ref{re4} statement, specifically on \eqref{f22} and \eqref{Map}, and \eqref{av_aass}, $\mathrm{AA-SS}_\mathrm{max\ E}$ provides $\sim \frac{\beta\kappa}{\kappa+1}(0.85\sqrt{M}-0.64)$ more energy units than $\mathrm{AA-SS}_\mathrm{min\ var}$ on average.}. In fact, the performance curve of $\mathrm{AA-SS}_\mathrm{max\ E}$ lies approximately in the middle between the curves of $\mathrm{AA-SS}_\mathrm{min\ var}$ and the idealistic $\mathrm{AA-SS}$ ($\phi=0$). Meanwhile, in terms of variance and energy outage probability they do not differ significantly. Regarding $\mathrm{AA-IS}$ and $\mathrm{SA}$, we can observe that they even outperform $\mathrm{AA-SS}_\mathrm{min\ var}$'s performance, while $\mathrm{AA-SS}_\mathrm{max\ E}$ is strictly superior in the region of large path loss. Therefore, far devices are the most benefited from $\mathrm{AA-SS}_\mathrm{max\ E}$. Meanwhile, since the variance of the incident RF energy (and also harvested energy as observed in Fig.~\ref{Fig6}b) is the lowest under $\mathrm{AA-IS}$ and $\mathrm{SA}$, they are capable of providing a more stable energy supply, which when operating near (or in) saturation makes the devices to harvest more energy on average when  compared to $\mathrm{AA-SS}_\mathrm{max\ E}$. That is why the average harvested energy under $\mathrm{AA-IS}$ and $\mathrm{SA}$ is the greatest among the practical schemes for $\beta>6$ dBm. Specifically, $\mathrm{AA-IS}$ performs the best in that region, while $\mathrm{SA}$ is slightly better than $\mathrm{AA-IS}$ for $\beta<4$ dBm. Discussions on this were carried out already in Subsections~\ref{SA} and \ref{ehs}.

\begin{figure}[t!]
	\centering  
	\includegraphics[width=0.9\columnwidth]{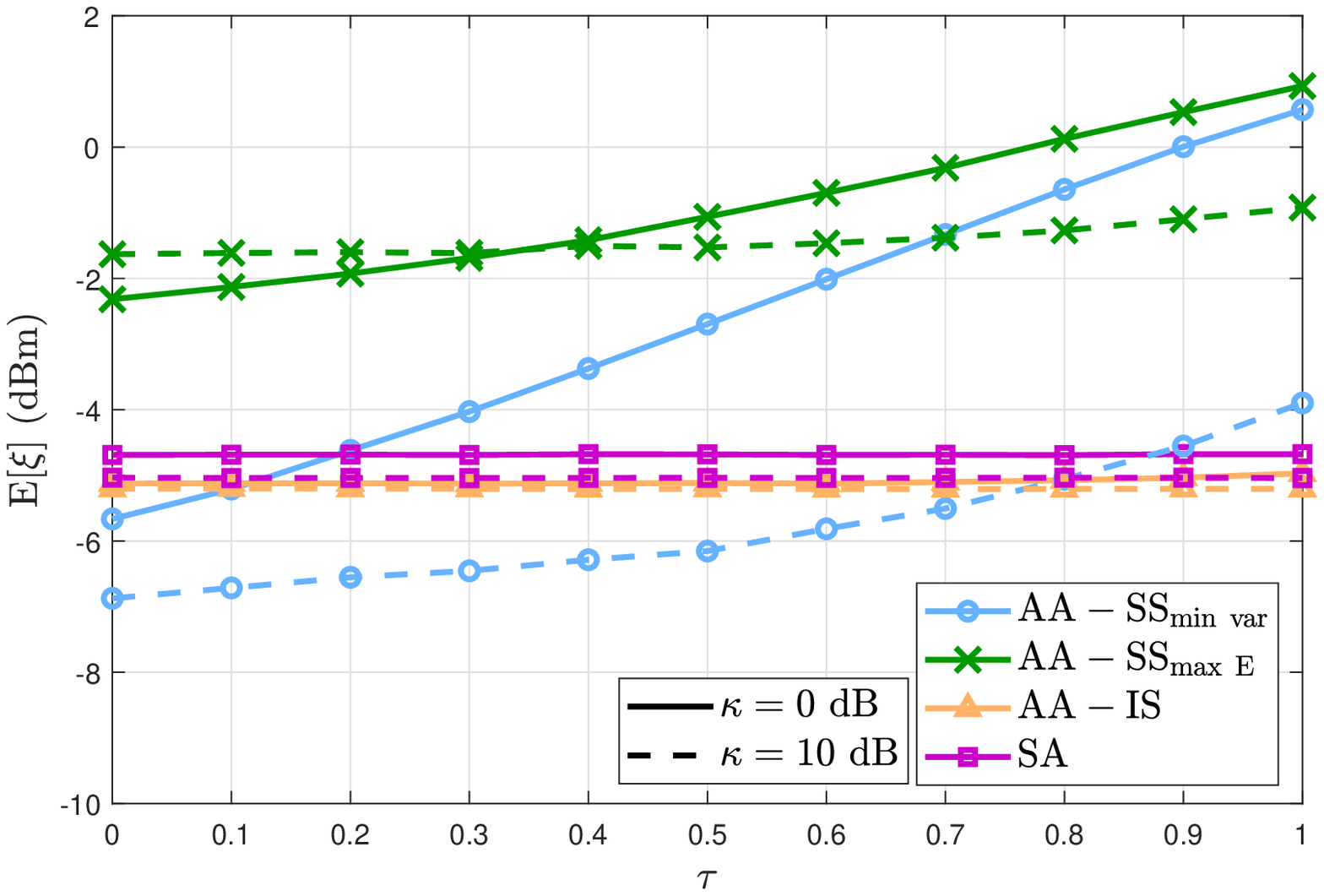}\\
	\includegraphics[width=0.9\columnwidth]{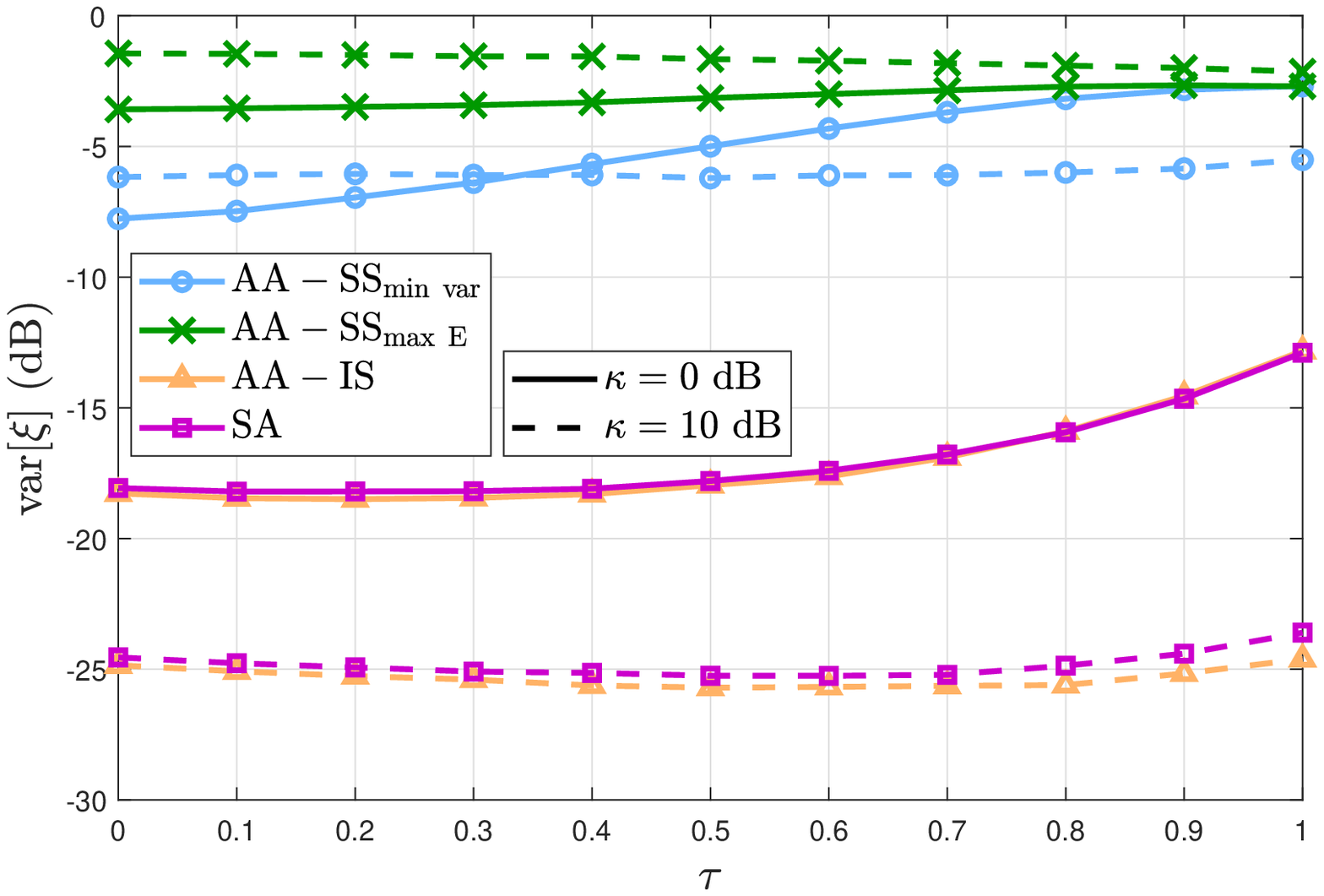}\\	
	\includegraphics[width=0.9\columnwidth]{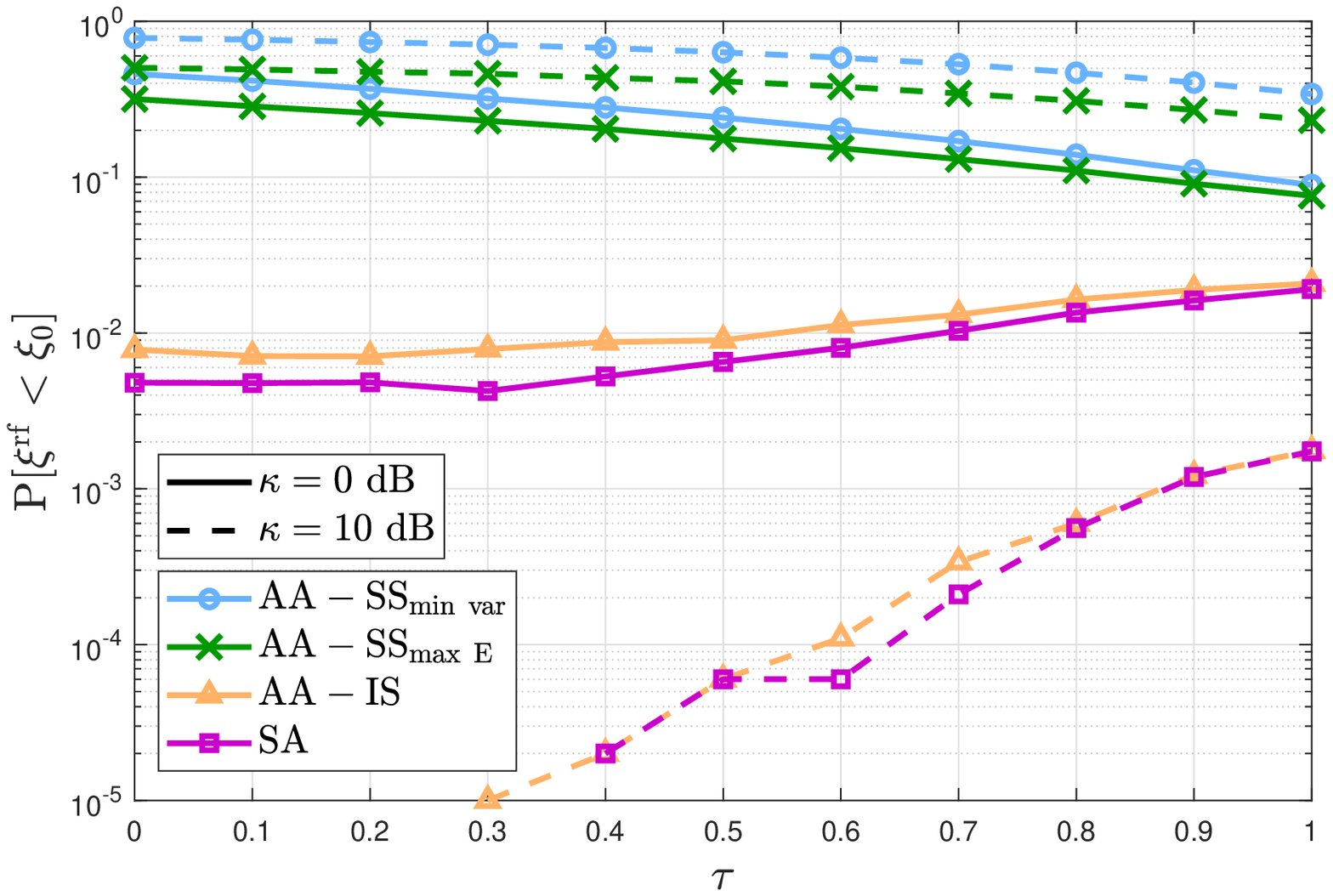}
	\caption{$a)$ Average harvested energy (top), $b)$ variance of the harvested energy (middle), and $c)$ energy outage probability (bottom), as a function of $\tau$ for $\kappa=\{0,10\}$ dB and $\beta=2$ dBm.}	
	\vspace{-2mm}
	\label{Fig7}
\end{figure}
As we already demonstrated and verified the necessity of considering the mean phase shifts in the analysis, from now on we just focus on showing the performance results under random mean phase shifts, e.g. we dispense of ``scheme'' ($\phi=0$) performance curves. 

\begin{figure}[t!]
	\centering  
	\includegraphics[width=0.9\columnwidth]{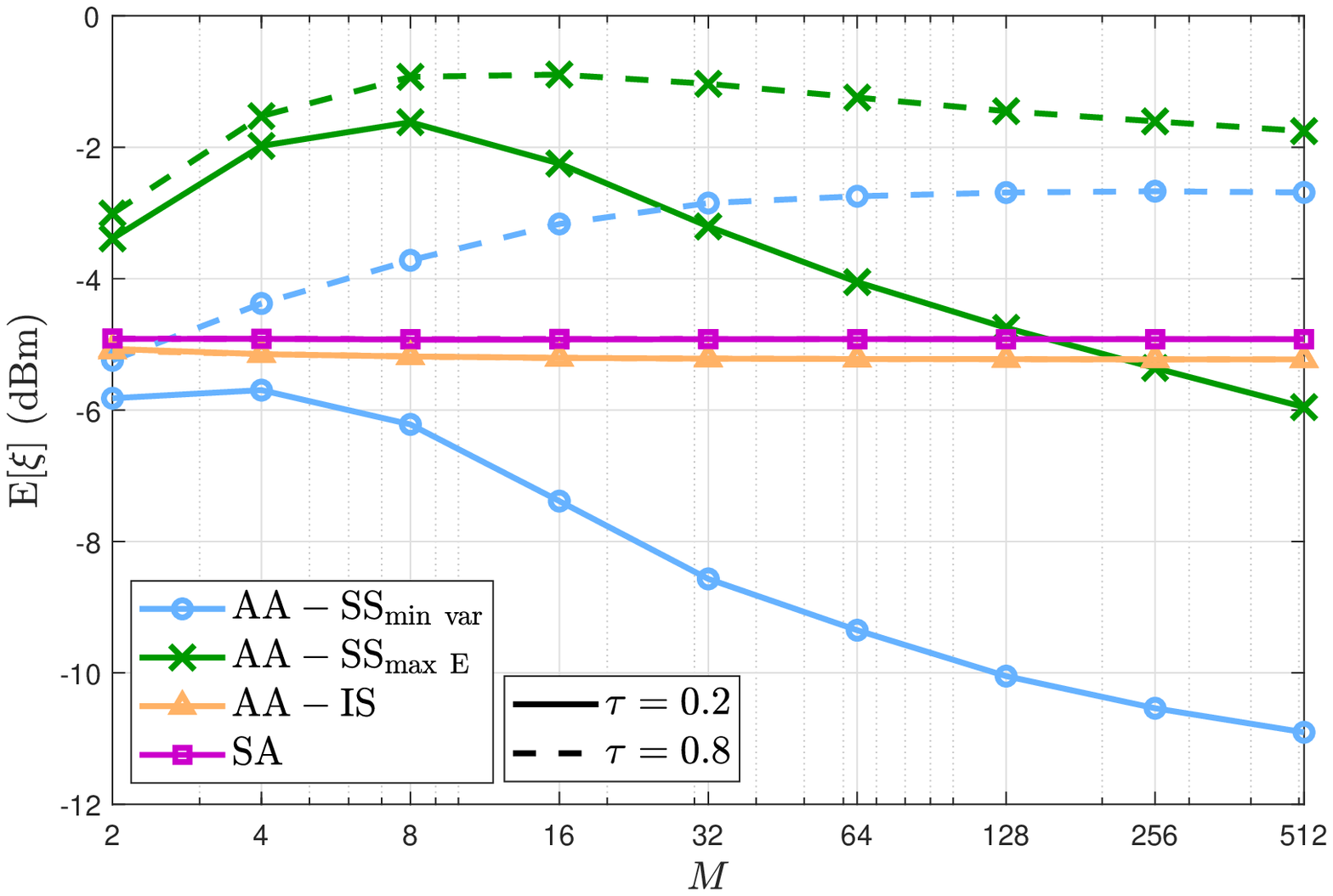}\\
	\includegraphics[width=0.9\columnwidth]{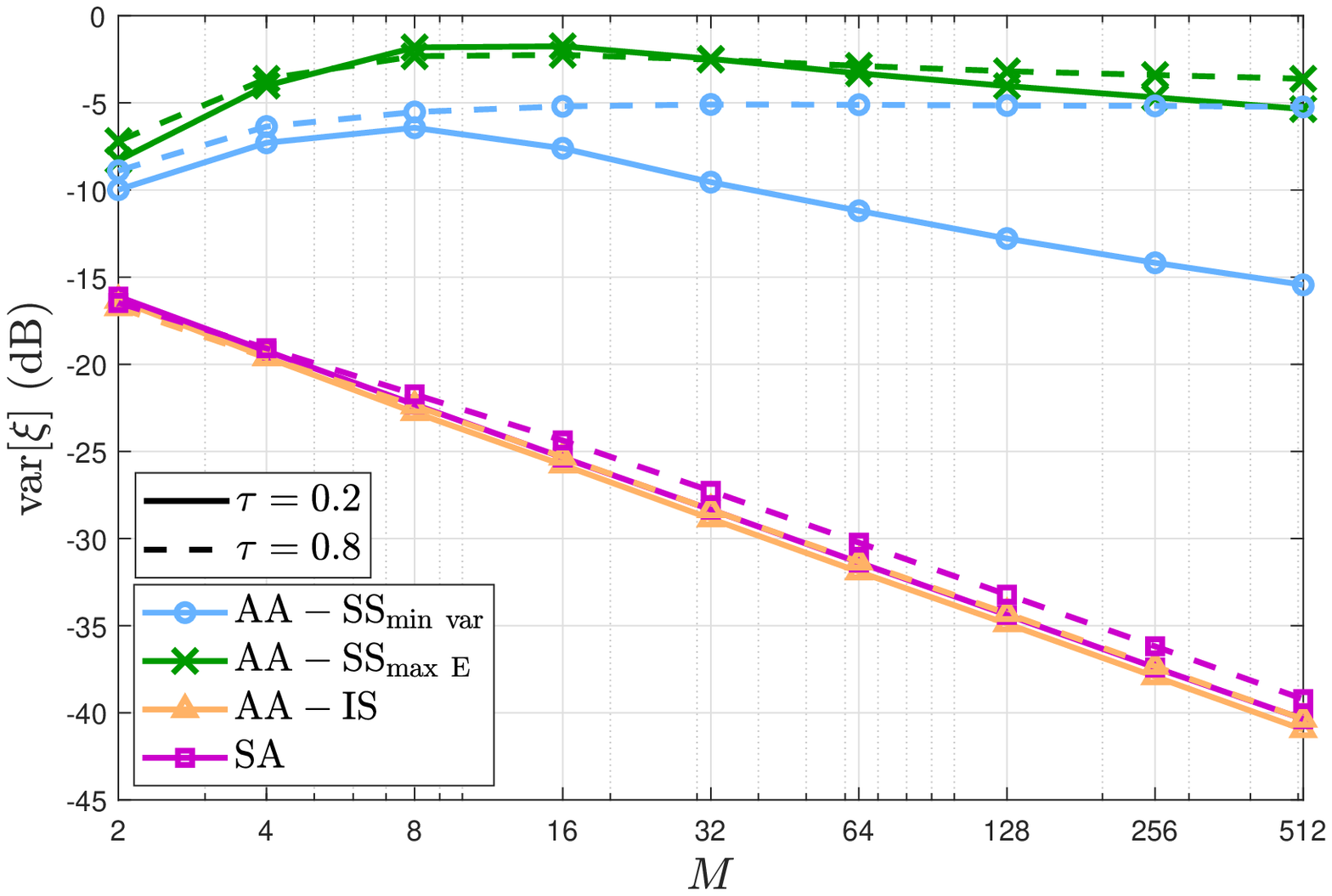}\\	
	\includegraphics[width=0.9\columnwidth]{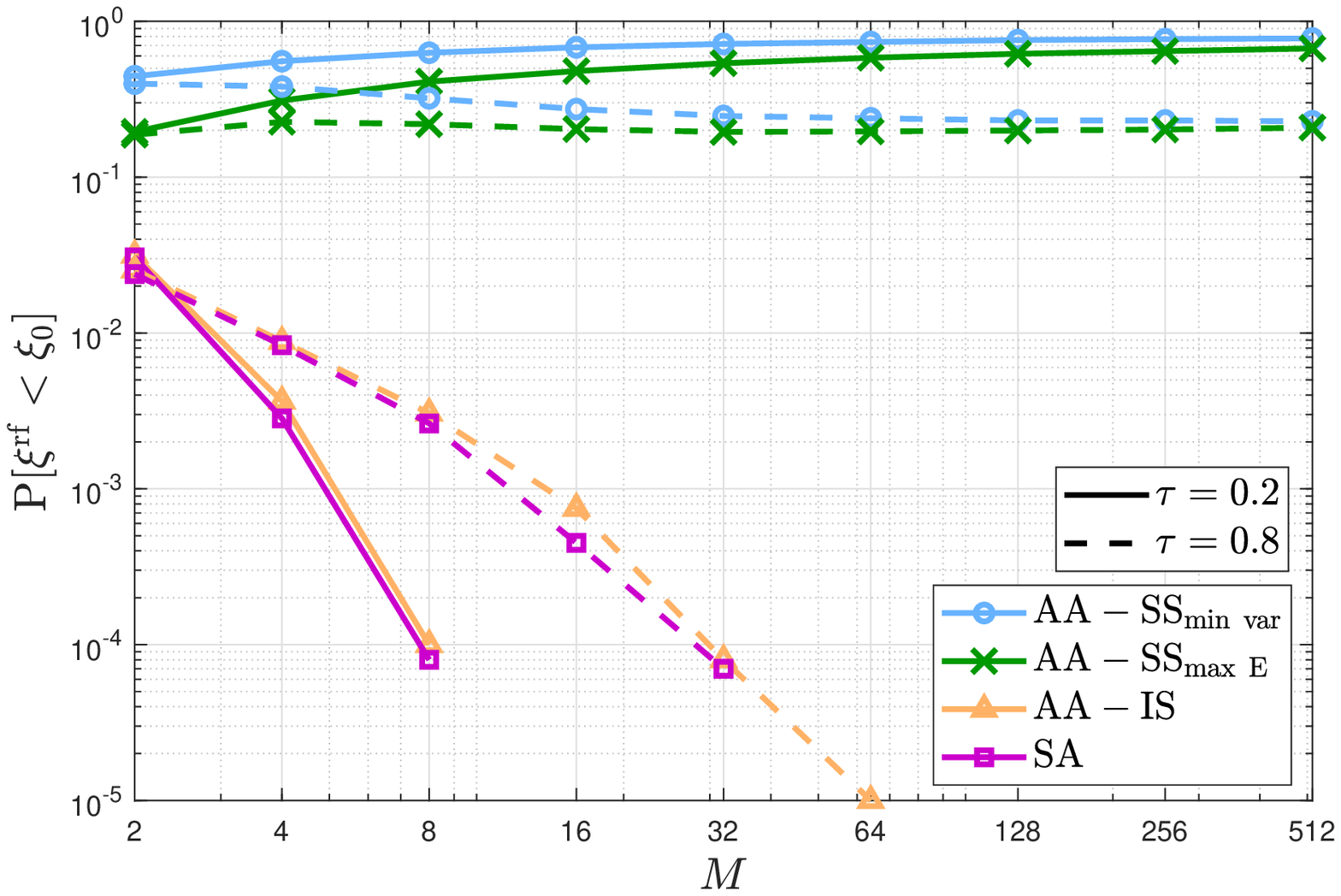}
	\caption{$a)$ Average harvested energy (top), $b)$ variance of the harvested energy (middle), and $c)$ energy outage probability (bottom), as a function of $M$ for $\tau=\{0.2,0.8\}$ and $\beta=2$ dBm.}		\vspace{-2mm}
	\label{Fig8}
\end{figure}
In Fig.~\ref{Fig7} we show the statistics (average in Fig.~\ref{Fig7}a and variance in Fig.~\ref{Fig7}b) of the harvested energy, and the energy outage probability (Fig.~\ref{Fig7}c), as a function of the exponential correlation coefficient for two LOS profiles. As evidenced, a greater correlation is beneficial under the $\mathrm{AA-SS}$ schemes, especially under poor LOS where channels are more random. A very counter-intuitive result is that LOS is not always beneficial when transmitting the same signal simultaneously through all antennas. In fact, the energy outage probability under such schemes is lower for $\kappa=0$ dB than for $\kappa=10$ dB, which also holds in case of the average harvested energy when preventive phase shifting is not used, e.g. $\mathrm{AA-SS}_\mathrm{min\ var}$, and for $\tau>0.35$ in case of $\mathrm{AA-SS}_\mathrm{max\ E}$. 
Meanwhile, since correlation (LOS) is well-known to decrease (increase) the diversity, $\mathrm{AA-IS}$ and $\mathrm{SA}$ schemes are affected by (benefited from) an increasing $\tau$ ($\kappa$) as shown in Fig.~\ref{Fig7}b and Fig.~\ref{Fig7}c. Notice however that as claimed in Remark~\ref{re6} the average statistics of the energy harvested under such schemes is not affected in any way by the correlation, and only a bit by the LOS factor due to the non-linearity of the energy harvester.
In general, $\mathrm{AA-SS}_\mathrm{max\ E}$ is the clear winner in terms of maximum average harvested energy, while $\mathrm{AA-IS}$ and $\mathrm{SA}$ schemes offer better performance in terms of energy outage.
\subsection{On the impact of the number of antennas}
\begin{figure}[t!]
	\centering  
	\includegraphics[width=0.9\columnwidth]{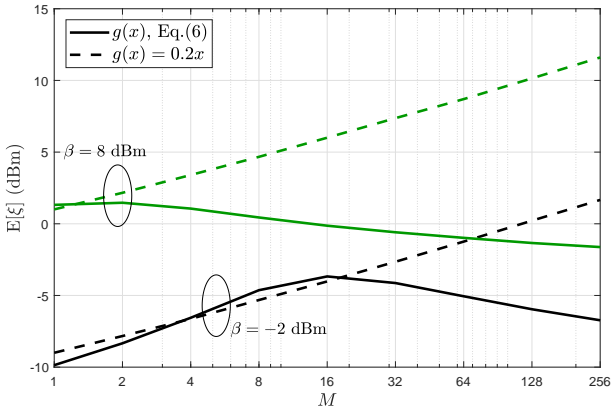}
	\caption{Average harvested energy as a function of $M$ for $\beta=\{-2,8\}$ dBm. Non-linear EH model vs linear EH model with $\eta=20\%$.}		
	\label{Fig9}
	\vspace{-3mm}
\end{figure}
Herein we investigate the impact of the number of transmit antennas on the system performance. Fig.~\ref{Fig8} show the average (Fig.~\ref{Fig8}a) and variance (Fig.~\ref{Fig8}b) of the harvested energy, and the energy outage probability (Fig.~\ref{Fig8}c) for $\tau=\{0.2,0.8\}$. As expected, the average harvested energy under $\mathrm{AA-IS}$ and $\mathrm{SA}$ remains constant as $M$ increases, while the variance and consequently the energy outage probability decrease. The latter improves as $\tau$ decreases as commented in previous subsection. Meanwhile, it is observed that a greater $M$ not always allows improvements on the system performance when transmitting the same signal simultaneously through all antennas. Interestingly, as the correlation increases a greater number of antennas is advisable, which does not occur for small $\tau$. 

\begin{figure}[t!]
	\centering  
	\ \includegraphics[width=1\columnwidth]{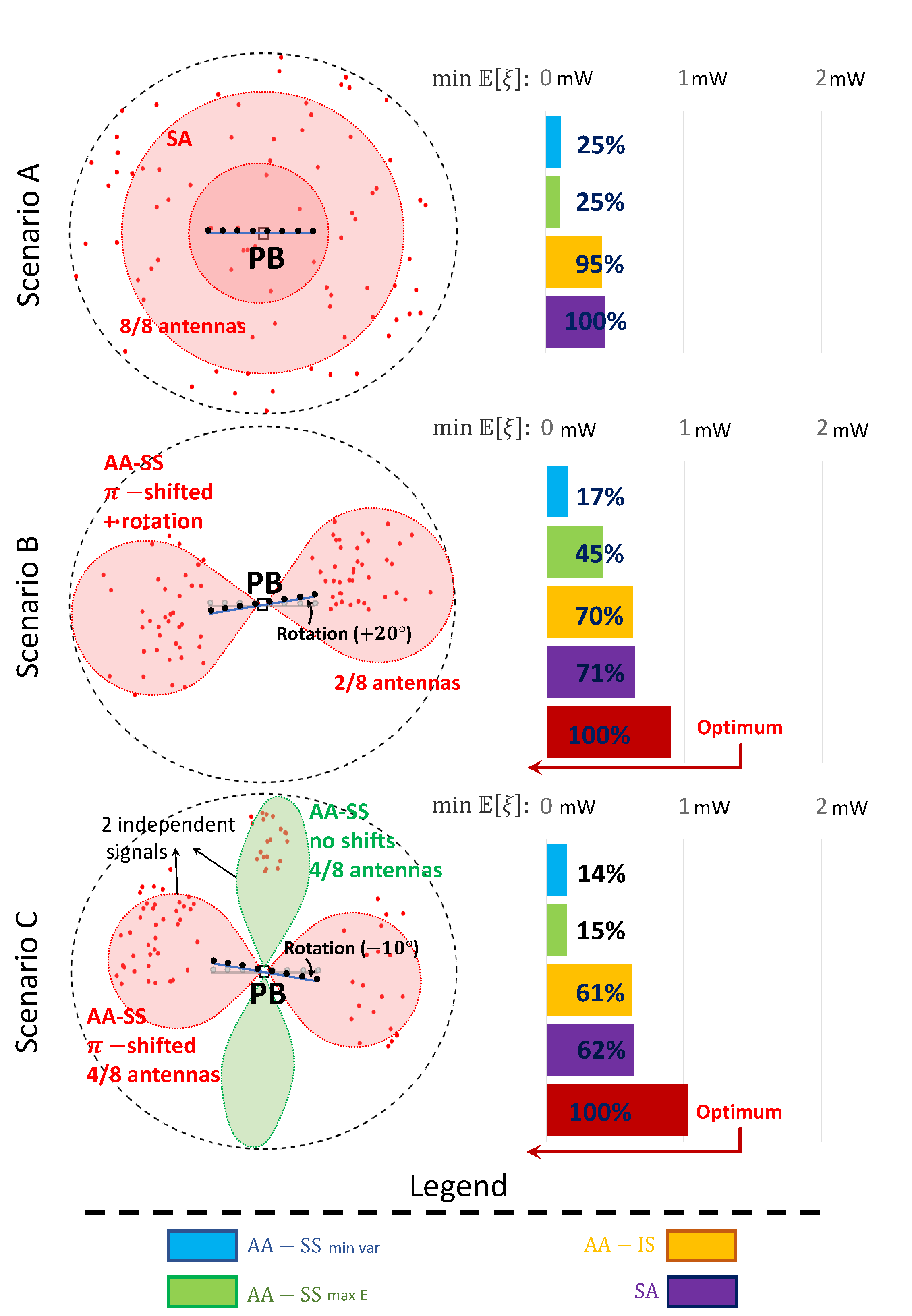}
	\vspace{-1mm}
	\caption{Devices' minimum harvested energy for three different scenarios: A) devices distributed uniformly in the area, B) devices distributed in two clusters, and C) devices distributed in three clusters. The optimum powering scheme is illustrated together with each scenario. The shown radiation patterns are by no means exhaustive and are included for reference only.}	
	\vspace{-3mm}
	\label{FigSC}
\end{figure}
Results in Fig.~\ref{Fig8} evidence that there is an optimum number of transmit antennas $M^*$. For instance, in terms of average harvested energy, $M^*=4$ ($8$) and $M^*=256$ ($16$) for $\tau=0.2$ and $\tau=0.8$, respectively, and under the operation of $\mathrm{AA-SS}_\mathrm{min\ var}$  ($\mathrm{AA-SS}_\mathrm{max\ E}$). 
Such phenomenon is mostly due to the fact that as $M$ increases, the chances of operating close to the minima of $f(\bm{\psi},\theta)$ increase for a random $\theta$ (see Remark~\ref{rr2}), while the maximums can not be fully exploited due to the non-linearity of the EH circuitry, specifically due to the saturation phenomenon.
To deepen this, we show in Fig.~\ref{Fig9} the average harvested energy as a function of $M$ for two different path loss profiles, while comparing the performance under the non-linear EH function to the one attained under the ideal linear EH model with $20\%$ of PTE\footnote{Instead of the maximum PTE we utilized a more conservative value.}. Notice that the performance grows unbounded under such ideal model, which also happens when analyzing the average incident RF power. However, under the non-linearities of the energy harvester that is not longer the case. Since for $\beta=-6$ dBm the device is likely to be operating close to saturation, the optimum number of antennas is small, e.g. $M^*=1$ or $2$, but under greater path loss, such number goes up to $M^*=16$. 
\subsection{Exploiting positioning information}\label{position}
Throughout the paper we have assumed that the EH devices are distributed uniformly around the PB, either because their locations are unknown, or their deployment distribution is nearly homogeneous. Herein, we show that even when such assumptions do not hold we can still benefit from the analyzed CSI-free powering schemes. As performance metric we adopt the devices' minimum average harvested energy, e.g., worst node's performance, hence, we aim to reveal the max-min fairest scheme for wirelessly powering the three example scenarios depicted in Fig.~\ref{FigSC}.
For all them, we assume a PB serving a $10$m$-$radius circular area where 80 EH devices\footnote{Notice that WET is usually practical at a scale of a few meters to tens of meters, thus, a $10$m$-$radius area is a valid assumption. Also, the projections towards 6G point to challenging scenarios with up to $10$ devices/$\mathrm{m}^2$ \cite{Mahmood.2019,Mahmood.2020}, thus, the considered $80$ deployed EH devices do not constitute a very large number compared to the what is expected 10 years from now. Obviously, the larger the number of EH devices is, the more beneficial the analyzed CSI-free schemes are when compared to the traditional CSI-based schemes \cite{Lopez.2019}.} are i) Scenario A: distributed uniformly in the area, ii) Scenario B: distributed in two clusters on opposite sides of the PB, and iii) Scenario C: distributed in three clusters, two clusters as in Scenario B plus another cluster shifted $\sim 90^\circ$ with respect to the previous. Finally, we model the average RF power available at certain distance $d$ from the PB as $\beta\ (\mathrm{dBm})=30-27\log_{10} d(\mathrm{m})$, which may correspond to setups where the PB's total transmit power is 1W, and channels are subject to log-distance path losses with exponent 2.7.

As evidenced by Scenario A, $\mathrm{SA}$ constitutes the fairest scheme when the nodes are distributed over the whole area, e.g., without obvious clustering patterns. Notice that although $\mathrm{AA-IS}$ attains a similar performance, $\mathrm{SA}$ is preferable since it favors the farthest devices with EH hardware operating close to the sensitivity region. Meanwhile, $\mathrm{AA-SS}_\mathrm{min\ var}$ and $\mathrm{AA-SS}_\mathrm{max\ E}$ are in clear disadvantage because of their performance minima at different angular directions as illustrated in Fig.~\ref{Fig23}. Conversely, if the devices are deployed as in Scenario B, the optimum strategy requires to rotate the PB's antenna array\footnote{This may be possible in static setups, where the task is committed to the technician/user, or in slow-varying environments, where the PB itself is equipped with rotary-motor skills and it is capable of adjusting its orientation.} $20^\circ$ counter-clockwise and then use $\mathrm{AA-SS}_\mathrm{max\ E}$.
Notice that this is more convenient than for instance rotating the antenna array $70^\circ$ clockwise and using $\mathrm{AA-SS}_\mathrm{min\ var}$ since this scheme provides narrower beams and thus powering more dispersed clusters becomes less efficient.
Also, just  2 (consecutive) out of the 8 available antennas are utilized such that the side-beams are sufficiently wide to efficiently cover all the EH devices. By doing so, the worst performing device can harvest up to 1.5 dB more energy than under the $\mathrm{SA}$ scheme. Similarly, the optimum strategy in Scenario C demands a $10^\circ$ clockwise rotation of the PB's antenna array. Meanwhile, since the devices are now grouped into clusters at $\sim 0^\circ$, $\sim 90^\circ$ and $\sim 180^\circ$, 2 independent signals can generate the required beams to power them, as illustrated in Fig.~\ref{FigSC}. Each signal is transmitted with equal power over 4 (consecutive) out of the 8 antennas. The signals are transmitted respectively according to the schemes $\mathrm{AA-SS}_\mathrm{min\ var}$ and $\mathrm{AA-SS}_\mathrm{max\ E}$ such that all the clusters can be efficiently reached. 
By doing so, the worst performing device can harvest up to 2 dB more energy than under the $\mathrm{SA}$ scheme.

All WET schemes analyzed in this paper aim at  powering massive deployments of EH devices without costly CSI acquisition overheads. Results suggest that:
\begin{itemize}
    \item $\mathrm{SA}$ and $\mathrm{AA-IS}$ are preferable when the devices' deployments are not clustered, e.g., for powering dense deployments of sensors and RFIDs in warehouse's storage areas. As highlighted in Remark~\ref{re1} and discussed in previous sections, the specific choice depends on the EH operation region of the critical EH devices, e.g. devices farthest from the PB or those with more stringent energy demands;
    \item $\mathrm{AA-SS}$ is preferable when devices are clustered in specific spatial directions; e.g., for powering parking lot sensors. 
\end{itemize}
\section{Conclusion}\label{conclusions}
In this paper, we analyzed CSI-free schemes that a dedicated multi-antenna power station can utilize when powering wirelessly a large set of single-antenna devices. Differently from our early work \cite{Lopez.2019}, such CSI-free schemes were studied in a more practical setup that takes into account the mean phase shifts between antenna elements. In addition to $\mathrm{AA-SS}$ (All Antennas transmitting the Same Signal) and $\mathrm{SA}$ (Switching Antennas) \cite{Lopez.2019} schemes, we analyzed the $\mathrm{AA-IS}$ (All Antennas transmitting Independent Signals) scheme as well. We demonstrated that those devices far from the Power Beacon (PB) and more likely to operate near their sensitivity level, benefit more from the $\mathrm{SA}$ scheme than from $\mathrm{AA-IS}$. However, those closer to the PB and more likely to operate near saturation, benefit more from $\mathrm{AA-IS}$.

We characterized the distribution of the RF energy at the EH receiver in correlated Rician fading channels under each WET scheme, and found out that while $\mathrm{AA-IS}$ and $\mathrm{SA}$ cannot take advantage of the multiple antennas to improve the average statistics of the incident RF power, the energy dispersion can be significantly reduced. Meanwhile, the gains in terms of average RF energy depend critically on the mean phase shifts between the antenna elements when using $\mathrm{AA-SS}$. 
In that regard, we show the considerable performance gaps between the idealistic $\mathrm{AA-SS}$ studied in \cite{Lopez.2019} and this scheme when considering channels with different mean phases. Even under such performance degradation $\mathrm{AA-SS}$ still provides the greatest average RF energy when compared to $\mathrm{AA-IS}$ and $\mathrm{SA}$, although its associated energy outage probability is generally the worst.
Those trade-offs make $\mathrm{AA-IS}$ and $\mathrm{SA}$ schemes suitable for powering devices under harvest-then-transmit/cooperate -like protocols, while $\mathrm{AA-SS}$ seems more appropriate for scenarios where the IoT devices are allowed to accumulate energy.

Additionally, we attained the preventive phase shifting that the power station could utilize for maximizing the average energy delivery or minimizing its dispersion for each of the schemes. Specifically, consecutive antennas must be $\pi$ phase-shifted for optimum average energy performance under $\mathrm{AA-SS}$, while under other optimization criteria and/or different schemes, there is no need of carrying out any preventive phase shifting. 
Numerical results evidenced that correlation is beneficial under $\mathrm{AA-SS}$, while a very counter-intuitive finding was that a greater LOS and/or number of antennas was not always beneficial under such WET scheme. 	 
Meanwhile, both $\mathrm{AA-IS}$ and $\mathrm{SA}$ schemes benefit from small correlation, and large number of antennas and LOS. 
Finally, we showed that $\mathrm{AA-SS}$ ($\mathrm{SA}$ and $\mathrm{AA-IS}$) is (are) the fairest when devices are (are not) clustered in specific spatial directions.
All these are fundamental results that can be used when designing practical WET systems.

\appendices 

\section{Proof of Theorem~\ref{the1}}\label{App-A}
Departing from \eqref{aassf} and using \eqref{hxy}, the RF power available as input to the energy harvester obeys
\begin{align}
\xi^{\mathrm{rf}}_{\mathrm{aa-ss}}&=\frac{\beta}{M}\big|\bm{1}^T\mathbf{h}_{x}+\mathbbm{i}\bm{1}^T\mathbf{h}_{y}\big|^2=\frac{\beta}{M}\Big[\big(\bm{1}^T\mathbf{h}_{x}\big)^2+\big(\bm{1}^T\mathbf{h}_{y}\big)^2\Big]\nonumber\\
&\stackrel{(a)}{=}\frac{\beta R_{\scriptscriptstyle\sum}}{2(\kappa+1)M}\big(\theta_x^2+\theta_y^2\big)\nonumber\\
&\stackrel{(b)}{\sim}\!\!\frac{\beta R_{\scriptscriptstyle\sum}}{2(\kappa\!+\!1)M}\chi^2\Big(\!2,\!\frac{2\kappa\big(\upsilon_1(\bm{\psi},\phi)^2\!+\!\upsilon_2(\bm{\psi},\phi)^2\big)}{R_{\scriptscriptstyle\sum}}\!\Big), 
\end{align}
where $R_{\scriptscriptstyle\sum}=\bm{1}^T\mathbf{R}\ \!\bm{1}$, while $\upsilon_1$ and $\upsilon_2$ are given in \eqref{v1v2}.
Notice that $(a)$ comes from setting $\theta_{x,y}^2=\frac{2(\kappa+1)}{ R_{\scriptscriptstyle\sum}}\big(\bm{1}^T\mathbf{h}_{x,y}\big)^2$, where $\theta_{x,y}\sim \mathcal{N}\Big(\sqrt{\frac{\kappa}{ R_{\scriptscriptstyle\sum}}}(\upsilon_1\pm\upsilon_2,1\Big)$ since $\bm{1}^T\mathbf{h}_{x,y}$ is still a Gaussian RV, where the mean is equal to the sum of the means of each element of vector $\mathbf{h}_{x,y}$, while the variance equals the sum of the elements of covariance matrix $\frac{1}{2(\kappa+1)}\mathbf{R}$ \cite{Novosyolov.2006}. Therefore, the mean and variances are $\sqrt{\frac{\kappa}{2(\kappa+1)}}(\upsilon_1\mp\upsilon_2)$ and $\frac{ R_{\scriptscriptstyle\sum}}{2(\kappa+1)}$, respectively. 
Then, $(b)$ comes from using the definition of a non-central chi-squared RV \cite[Cap.2]{Proakis.2001} along with simple algebraic transformations. Finally,  after using \eqref{fv} we attain \eqref{AA-SS}.  \hfill 	\qedsymbol
\section{Proof of Theorem~\ref{the2}}\label{App-B}
Based on \eqref{aaisf}, the RF energy available at $S$ under the $\mathrm{AA-IS}$ operation is given by
\begin{align}
\xi_\mathrm{aa-is}^\mathrm{rf}&=\frac{\beta}{M}\big(\mathbf{h}_x^T\mathbf{h}_x+\mathbf{h}_y^T\mathbf{h}_y\big),\label{xi_aa-is} 
\end{align}
for which we analyze its distribution as follows.
Let us define $\mathbf{z}_{x,y}\!=\!\sqrt{2(\kappa+1)}\mathbf{R}^{-1/2}\Big(\mathbf{h}_{x,y}\!-\!\sqrt{\frac{\kappa}{2(\kappa\!+\!1)}}\bm{\omega}_{x,y}\Big)$ which is distributed as $\mathcal{N}\big(\mathbf{0},\mathbf{I}\big)$, then
\begin{align}
\mathbf{h}_{x,y}&=\frac{1}{\sqrt{2(\kappa+1)}}\mathbf{R}^{1/2}\mathbf{z}_{x,y}+\sqrt{\frac{\kappa}{2(\kappa+1)}}\bm{\omega}_{x,y}
\end{align}
\begin{align}
&\mathbf{h}_{x,y}^T\mathbf{h}_{x,y}\nonumber\\
&=\frac{1}{2(\kappa+1)}\big(\mathbf{R}^{1/2}\mathbf{z}_{x,y}+\sqrt{\kappa}\bm{\omega}_{x,y}\big)^T\big(\mathbf{R}^{1/2}\mathbf{z}_{x,y}+\sqrt{\kappa}\bm{\omega}_{x,y}\big)\nonumber\\
&=\!\frac{1}{2(\kappa\!+\!1)}\!\big(\mathbf{z}_{x,y}\!+\!\mathbf{R}^{-1/2}\!\sqrt{\kappa}\bm{\omega}_{x,y}\big)^{\!T}\!\mathbf{R}\big(\mathbf{z}_{x,y}\!+\!\mathbf{R}^{-1/2}\!\sqrt{\kappa}\bm{\omega}_{x,y}\big),\label{hh}
\end{align}
where last step comes from simple algebraic transformations. Notice that
$\mathbf{R}=\mathbf{Q}\mathbf{\Lambda}\mathbf{Q}^T$
is the spectral decomposition of $\mathbf{R}$ \cite[Ch.21]{Harville.2008}, where $\mathbf{\Lambda}$ is a diagonal matrix containing the eigenvalues of $\mathbf{R}$, and $\mathbf{Q}$ is a matrix whose column vectors are the orthogonalized eigenvectors of $\mathbf{R}$. In order to find the eigenvalues, $\{\lambda_j\}$, of $\mathbf{R}$, we require solving $\mathrm{det}\big(\mathbf{R}-\lambda\mathbf{I}\big)=0$ for $\lambda$, which is analytical intractable for a general matrix $\mathbf{R}$. 
However, there is analytical tractability for the special case of uniform spatial correlation. Notice that for the special case of uniformly spatial correlated fading such that the antenna elements are correlated between each other with coefficient $\rho$ ($R_{i,j}=\rho,\ \forall i\ne j$), we have that
	\begin{align}\label{delta2}
	R_{\scriptscriptstyle\sum}=M\big(1+(M-1)\rho\big).
	\end{align}
	Also, in order to guarantee that $\mathbf{R}$ is positive semidefinite and consequently a viable covariance matrix, $\rho$ is lower bounded by $-\frac{1}{M-1}$ \cite{Chen.2015}, thus $-\frac{1}{M-1}\le \rho\le 1$ and $0\le R_{\scriptscriptstyle\sum}\le M^2$. Then, under such correlation model  we have that \cite[Eq.(32)]{Lopez.2019}
\begin{align}\label{L}
\mathbf{\Lambda}=\mathrm{diag}\big(1-\rho,\cdots,1-\rho,1+(M-1)\rho\big),
\end{align}
\begin{figure*}[th!]	
	\begin{align}\label{PP}
	\mathbf{Q}^T=\left[\begin{smallmatrix}
	-\frac{1}{\sqrt{1\times 2}}& 0 &  0 &\ldots& 0& 0 & 0 & \frac{1}{\sqrt{1\times 2}}\\ 
	-\frac{1}{\sqrt{2\times 3}} & 0& 0 &\ldots &0 & 0& \frac{2}{\sqrt{2\times 3}} & -\frac{1}{\sqrt{2\times 3}}\\
	-\frac{1}{\sqrt{3\times 4}} & 0& 0 &\ldots &0 & \frac{3}{\sqrt{3\times 4}}& -\frac{1}{\sqrt{3\times 4}} & -\frac{1}{\sqrt{3\times 4}}\\
	-\frac{1}{\sqrt{4\times 5}} & 0& 0 &\ldots &\frac{4}{\sqrt{4\times 5}} & -\frac{1}{\sqrt{4\times 5}}& -\frac{1}{\sqrt{4\times 4}} & -\frac{1}{\sqrt{4\times 5}}\\
	\vdots & \vdots& \vdots &\ddots &\vdots & \vdots& \vdots & \vdots\\
	-\frac{1}{\sqrt{(M-2)(M-1)}} & 0& \frac{M-2}{\sqrt{(M-2)(M-1)}} &\ldots &-\frac{1}{\sqrt{(M-2)(M-1)}} & -\frac{1}{\sqrt{(M-2)(M-1)}}& -\frac{1}{\sqrt{(M-2)(M-1)}} & -\frac{1}{\sqrt{(M-2)(M-1)}}\\
	-\frac{1}{\sqrt{(M-1)M}} & \frac{M-1}{\sqrt{(M-1)M}}& -\frac{1}{\sqrt{(M-1)M}} &\ldots &-\frac{1}{\sqrt{(M-1)M}} & -\frac{1}{\sqrt{(M-1)M}}& -\frac{1}{\sqrt{(M-1)M}} & -\frac{1}{\sqrt{(M-1)M}}\\
	\frac{1}{\sqrt{M}} & \frac{1}{\sqrt{M}}& \frac{1}{\sqrt{M}} & \ldots & \frac{1}{\sqrt{M}}&\frac{1}{\sqrt{M}}& \frac{1}{\sqrt{M}}& \frac{1}{\sqrt{M}}
	\end{smallmatrix}\right].\\
	\bottomrule\nonumber
	\end{align}
	\vspace{-8mm}
\end{figure*}
thus, matrix $\mathbf{R}$ is characterized by two different eigenvalues: $1-\rho$, which is independent of the number of antennas but has multiplicity $M-1$, and $1+(M-1)\rho$, whose multiplicity is $1$ but increases linearly with $M$.
Then, matrix $\mathbf{Q}^T$ can be written as shown at the top of the next page. Such representation can be verified for any $M$ by using specialized software for matrix processing such as Matlab or Wolfram.

Now, by substituting $\mathbf{R}=\mathbf{Q}\bm{\Lambda}\mathbf{Q}^T$ into \eqref{hh} while setting $\zeta_{x,y}\!=\!2(\kappa\!+\!1) \mathbf{h}_{x,y}^T\mathbf{h}_{x,y}$ yields
\begin{align}
\zeta_{x,y}&=\!\Big(\mathbf{z}_{x,y}\!+\!\sqrt{\kappa}\big(\mathbf{Q}\mathbf{\Lambda}\mathbf{Q}^T\!\big)^{\!-\frac{1}{2}}\bm{\omega}_{x,y}\Big)^T\!\mathbf{Q}\mathbf{\Lambda}\mathbf{Q}^T\!\times\nonumber\\
&\qquad\qquad\qquad\qquad\times\Big(\mathbf{z}_{x,y}\!+\!\sqrt{\kappa}\big(\mathbf{Q}\mathbf{\Lambda}\mathbf{Q}^T\!\big)^{\!-\frac{1}{2}}\bm{\omega}_{x,y}\Big)\nonumber\\
&\stackrel{(a)}{=}\Big(\mathbf{Q}^T\!\mathbf{z}_{x,y}\!+\!\sqrt{\kappa}\mathbf{\Lambda}^{\!-\frac{1}{2}}\mathbf{Q}^T\bm{\omega}_{x,y}\Big)^T\!\mathbf{\Lambda}\times\nonumber\\
&\qquad\qquad\qquad\qquad\times\Big(\mathbf{Q}^T\!\mathbf{z}_{x,y}\!+\!\sqrt{\kappa}\mathbf{\Lambda}^{\!-\frac{1}{2}}\mathbf{Q}^T\bm{\omega}_{x,y}\Big)\nonumber\\
&\stackrel{(b)}{=}\Big(\mathbf{P}_{x,y}+\sqrt{\kappa}\mathbf{u}_{x,y}\Big)^T\mathbf{\Lambda}\Big(\mathbf{P}_{x,y}+\sqrt{\kappa}\mathbf{u}_{x,y}\Big)\nonumber\\
&\stackrel{(c)}{=}\sum_{j=1}^{M}\lambda_{j}\Big(c_j+\sqrt{\kappa}{u_{\!x\!,y}}_j\Big)^2,\label{hhn}
\end{align}
where $(a)$ comes after some algebraic transformations, $(b)$ follows from taking $\mathbf{P}_{x,y}=\mathbf{Q}^T\mathbf{z}_{x,y}\sim\mathcal{N}\big(\mathbf{0},\mathbf{I}\big)$ and 
\begin{align}\label{d}
\mathbf{u}_{x,y}&=\mathbf{\Lambda}^{-\frac{1}{2}}\mathbf{Q}^T\bm{\omega}_{x,y}\nonumber\\
&=\mathbf{\Lambda}^{-\frac{1}{2}}\left[\begin{smallmatrix}
\frac{1}{\sqrt{2}}(\omega_M-\omega_1) \\
\frac{1}{\sqrt{2\times 3}}(2\omega_{M-1}-\omega_1-\omega_M) \\
\frac{1}{\sqrt{3\times 4}}\big(3\omega_{M-2}-\omega_1-\sum_{j=M-1}^{M}\omega_{j}\big) \\
\frac{1}{\sqrt{4\times 5}}\big(4\omega_{M-3}-\omega_1-\sum_{j=M-2}^{M}\omega_{j}\big) \\
\vdots \\
\frac{1}{\sqrt{(M-2)(M-1)}}\big((M-2)\omega_{3}-\omega_1-\sum_{j=4}^{M}\omega_{j}\big) \\
\frac{1}{\sqrt{(M-1)M}}\big((M-1)\omega_{2}-\omega_1-\sum_{j=3}^{M}\omega_{j}\big) \\
\frac{1}{\sqrt{M}}\sum_{j=1}^M \omega_j
\end{smallmatrix}\right],\\
{u_{\!x\!,y}}_j&\!=\!\left\{ \begin{array}{rl}
\frac{j\omega_{M-j+1}\!-\!\omega_1\!-\!\sum_{t=M-j+2}^{M}\omega_t}{\sqrt{j(j+1)\lambda_j}}, &   \mathrm{for}\ \  j\!\le\!M\!-\!1 \\
\frac{1}{\sqrt{M\lambda_M}}\sum_{t=1}^{M}\omega_t, &   \mathrm{for}\ \  j=M  
\end{array} 
\right.\!\!\!,\label{uxy}
\end{align}	
which come from using \eqref{L} and \eqref{PP}; 
while $(c)$ follows after taking $c_j\sim\mathcal{N}\big(0,1\big)$. 

By incorporating these results into \eqref{xi_aa-is} we obtain the distribution of $\xi_\mathrm{aa-is}^\mathrm{rf}$ as follows
\begin{align}\label{aa-is}
&\xi_\mathrm{aa-is}^\mathrm{rf}\sim \frac{\beta}{2M(\kappa+1)}(\zeta_x+\zeta_y)\nonumber\\
\ &\sim \frac{\beta}{2M(\kappa+1)}\sum_{j=1}^{M}\lambda_j\Big(\big(c_j+\sqrt{\kappa}{u_{\!x}}_j\big)^2+\big(\tilde{c}_j+\sqrt{\kappa}{u_{\!y}}_j\big)^2\Big)\nonumber\\
&\sim\frac{\beta}{2M(\kappa+1)}\bigg((1-\rho)\chi^2\Big(2(M-1),\frac{2\kappa\tilde{\upsilon}(\bm{\psi},\phi)}{1-\rho}\Big)+\nonumber\\
\ &\qquad +\big(1\!+\!(M\!-\!1)\rho\big)\chi^2\Big(2,\frac{\kappa f(\bm{\psi},\phi)}{M\big(1\!+\!(M\!-\!1)\rho\big)}\Big)\bigg) 
\end{align}
since $\tilde{c}_j\sim \mathcal{N}(0,1)$, 
\begin{align}
{u_{\!x}}_M^2+{u_{\!y}}_M^2&\!=\!\frac{1}{M\lambda_M}\bigg(\Big(\sum_{t=1}^{M}{\omega_x}_t\Big)^2+\Big(\sum_{t=1}^{M}{\omega_y}_t\Big)^2\bigg)\nonumber\\
&\!=\!\!\frac{\big(\upsilon_1(\bm{\psi},\phi)\!-\!\upsilon_2(\bm{\psi},\phi)\big)^2\!\!\!+\!\big(\upsilon_1(\bm{\psi},\phi)\!+\!\upsilon_2(\bm{\psi},\phi)\big)^2}{M\big(1+(M-1)\rho\big)}\nonumber\\
&\!=\!\frac{2f(\bm{\psi},\phi)}{M\big(1+(M-1)\rho\big)},
\end{align}
\begin{figure*}[th!]
\vspace{-6mm}
	\small
	\begin{align}
	\tilde{\upsilon}(\bm{\psi},\phi)&\stackrel{(a)}{=}\sum_{j=1}^{M-1}\frac{1}{2j(j+1)}\bigg(\Big(j{\omega_x}_{M-j+1}-{\omega_x}_1-\sum_{t=M-j+2}^{M}{\omega_x}_t\Big)^2+\Big(j{\omega_y}_{M-j+1}-{\omega_y}_1-\sum_{t=M-j+2}^{M}{\omega_y}_t\Big)^2\bigg)\nonumber\\
	&\stackrel{(b)}{=}\sum_{j=1}^{M-1}\frac{1}{2j(j+1)}\bigg(j^2\big({\omega_x}_{M-j+1}^2+{\omega_y}_{M-j+1}^2\big)+2+ \Big(\sum_{t=M-j+2}^{M}{\omega_x}_t\Big)^2+ \Big(\sum_{t=M-j+2}^{M}{\omega_y}_t\Big)^2-2j\big({\omega_x}_{M-j+1}+{\omega_y}_{M-j+1}\big)+\nonumber\\
	&\qquad\qquad\qquad\qquad\qquad\qquad\qquad\qquad -2j\Big({\omega_x}_{M-j+1}\sum_{t=M-j+2}^{M}{\omega_x}_t+{\omega_y}_{M-j+1}\sum_{t=M-j+2}^{M}{\omega_y}_t\Big)+2\sum_{t=M-j+2}^{M}\big({\omega_x}_t+{\omega_y}_t\big)\bigg)\nonumber\\
	&\stackrel{(c)}{=}\sum_{j=1}^{M-1}\frac{1}{j(j+1)}\bigg( \Big(\sum_{t=M-j+1}^{M-1}\cos\big(\psi_t+\Phi_t\big) \Big)^2+\Big(\sum_{t=M-j+1}^{M-1}\sin\big(\psi_t+\Phi_t\big) \Big)^2+j^2+1-2j\cos\big(\psi_{M-j}+\Phi_{M-j}\big)+\nonumber\\
	&\qquad-\!2j\Big(\!\cos\big(\psi_{M\!-j}+\Phi_{M\!-j}\big)\!\!\sum_{t\!=\!M\!-\!j\!+\!1}^{M-1}\!\!\cos \big(\psi_t+\Phi_t\big)\!+\!\sin\big(\psi_{M\!-j\!}+\Phi_{M\!-j\!}\big)\!\!\sum_{t\!=\!M\!-\!j\!+\!1}^{M-1}\!\!\sin \big(\psi_t+\Phi_t\big)\!\Big)\!+\!2\!\!\sum_{t\!=\!M\!-\!j\!+\!1}^{M-1}\!\!\cos \big(\psi_t+\Phi_t\big)\!\bigg).\label{upp1}\\
	\bottomrule\nonumber\vspace{-10mm}
	\end{align}
\end{figure*}
\begin{figure}[t!]
	\centering  \includegraphics[width=0.9\columnwidth]{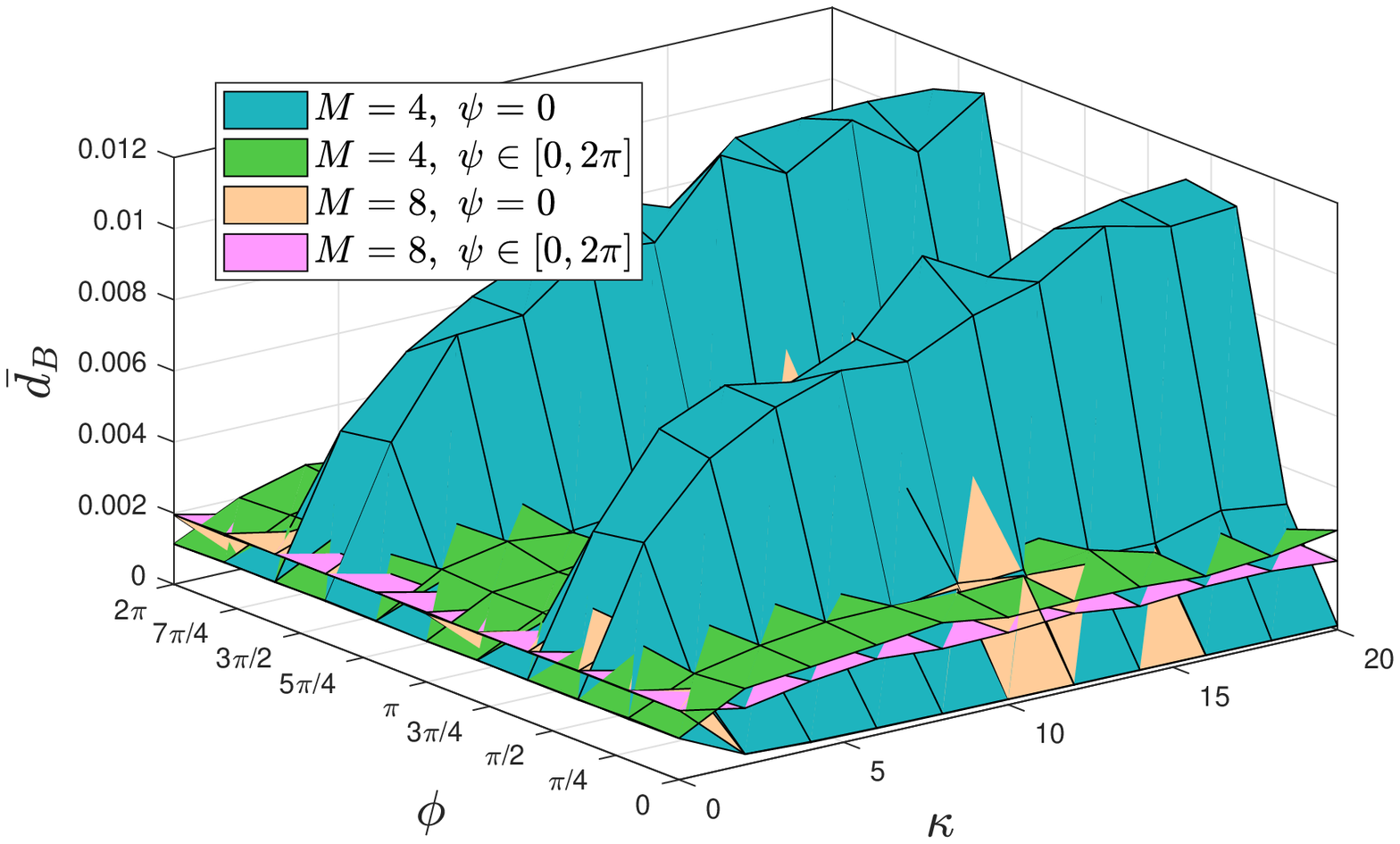}
	\vspace{-2mm}
   \includegraphics[width=1.05\columnwidth]{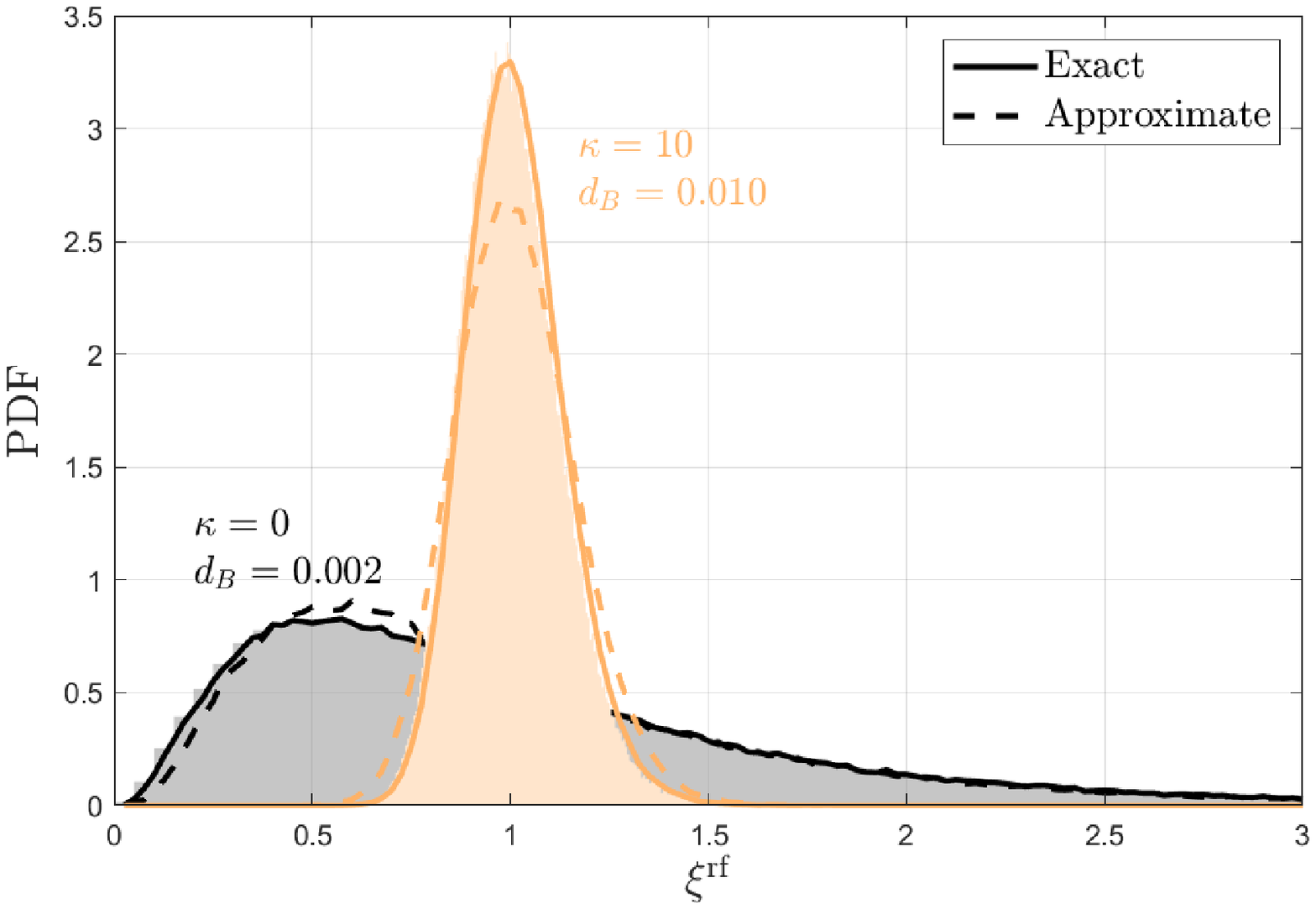}
	\vspace{-6mm}
	\caption{$a)$ Average Bhattacharyya distance between $\mathbf{\hat{p}_1}$ and $\mathbf{\hat{p}_2}$ as a function of $\kappa$ and $\phi$ for $M\in\{4,8\}$, $\bm{\psi}=\bm{0}$ (no preventive phase shifting) and $\bm{\psi}$ taken uniformly random from $[0,2\pi]^{M-1}$ (top). $b)$ Monte Carlo-based comparison between the exact $\mathbf{p_2}$ and proposed approximate $\mathbf{p_1}$ distributions of the incident RF power under the $\mathrm{AA-IS}$ scheme. 
		We set $M=4$, $\phi=3\pi/4$ and $\bm{\psi}=\bm{0}$. The comparison is carried out for two different cases: i) $\kappa=0$, $\bar{d}_B=0.002$; and ii) $\kappa=10$, $\bar{d}_B=0.010$ (bottom).}
	\label{FigApp}
\end{figure}
while $\tilde{\upsilon}=\frac{1}{2}\sum_{j=1}^{M-1}\lambda_{j}\big({u_{\!x}}_j^2+{u_{\!y}}_j^2\big)$, which appears expanded as a function of $\bm{\psi}$ and $\bm{\Phi}$ in \eqref{upp1} at the top of the next page and it is obtained by $(a)$ using \eqref{uxy}, $(b)$ expanding the quadratic binomials, and $(c)$ performing some algebraic simplifications by taking advantage of $\sin^2a+\cos^2a=1$.  Then, we attain \eqref{up1} by regrouping terms and performing further algebraic simplifications by taking advantage of $\cos a\cos b+\sin a\sin b=\cos(a-b)$.

Notice that \eqref{aa-is} holds under the assumption of uniform spatial correlation; however, given that for the $\mathrm{AA}$ scheme we were able of writing the distribution merely as a function of parameter $ R_{\scriptscriptstyle\sum}$, which is not linked to any specific kind of correlation, we can expect that the behavior under the $\mathrm{AA-IS}$ scheme depends, at least approximately, on $ R_{\scriptscriptstyle\sum}$ rather on the specific entries of matrix $\mathbf{R}$. In fact, such hypothesis has been shown to be accurate under the scenario of LOS components with equal mean phases studied in \cite{Lopez.2019}. To explore this, we substitute $\rho=\frac{ R_{\scriptscriptstyle\sum}-M}{M(M-1)}$ coming from \eqref{delta2}, into \eqref{aa-is}, such that we attain \eqref{aa-is2}.  \hfill 	\qedsymbol
\subsection{Validation}
 To evaluate the accuracy of \eqref{aa-is2} we utilize the Bhattacharyya distance metric \cite{Kailath.1967}, which measures the similarity of two probability distributions $\mathbf{p_1}$ and $\mathbf{p_2}$, and it is given by
\begin{align}\label{db}
d_B(\mathbf{p_1},\mathbf{p_2})&=-\ln\big(c_B(\mathbf{p_1},\mathbf{p_2})\big),
\end{align}
where $c_B$ is the Bhattacharyya coefficient \cite{Kailath.1967}. In the case of our interest, both probability distributions characterize $\xi_{\mathrm{aa-is}}^\mathrm{rf}$; however, for $\mathbf{p_1}$ we assume a uniform spatial correlation matrix $\mathbf{R}$, thus we may use \eqref{aa-is2} which is exact in such scenario; while for $\mathbf{p_2}$ we assume a randomly generated correlation matrix $\mathbf{R}^*$ such that $R_{\scriptscriptstyle\sum}=R_{\scriptscriptstyle\sum}^*$, thus, we evaluate the distribution directly from \eqref{xi_aa-is}. 
We utilize the histogram formulation for estimating $\mathbf{p_1}$ and $\mathbf{p_2}$. Specifically, we estimate $\mathbf{\hat{p}_1}=\{\hat{p}_{1,i}\}_{i=1,\cdots,m}$ (with $\sum_{i=1}^{m}\hat{p}_{1,i}=1$) and $\mathbf{\hat{p}_2}=\{\hat{p}_{2,i}\}_{i=1,\cdots,m}$ (with $\sum_{i=1}^{m}\hat{p}_{2,i}=1$), where $m$ is the number of histogram bins.  
Then, Bhattacharyya coefficient is calculated as \cite{Kailath.1967}
\begin{align}\label{cb}
c_B(\mathbf{\hat{p}_1},\mathbf{\hat{p}_2})=\sum_{i=1}^{m}\sqrt{\hat{p}_{1,i}\hat{p}_{2,i}}.
\end{align}
By substituting \eqref{cb} into \eqref{db} we calculate the similarity between distributions. Note that according to \eqref{aa-is2} we have that $d_B(\mathbf{p_1},\!\mathbf{p_2}\!)\in\![0,\infty\!)$, where $0$ corresponds to the case when $\mathbf{p_1}=\mathbf{p_2}$.

Fig.~\ref{FigApp}a shows the average Bhattacharyya distance, $\bar{d}_B$, as a function of $\kappa$ and $\phi$ for $M\in\{4,8\}$, $\bm{\psi}=\bm{0}$ (no preventive phase shifting) and $\bm{\psi}$ taken uniformly random from $[0,2\pi]^{M-1}$ to account for different possible preventive phase shiftings. We generated 1000 random correlation matrices and averaged over the Bhattacharyya distance for the distributions corresponding to each of them. 
We utilized $2\times 10^5$ distribution samples and set $m=240$, while the histogram edges were uniformly chosen between $0$ and $6$, which is an appropriate range since, without loss of generality, we used $\beta=1$. As observed in the figure, the largest difference between both distributions is when $\bm{\psi}=\bm{0}$,  $M$ is small, e.g. $M=4$, and $\phi\approx \pi/4 \pm \pi/2$. In such scenario $\bar{d}_B$ increases with $\kappa$; however, according to further extensive simulations although $d_B$ values become close to $0.012$ they never surpass such limit. Meanwhile, for other parameter configurations $\bar{d}_B$ becomes significantly smaller.

We select $\bar{d}_B=0.010$ and $\bar{d}_B=0.002$ to illustrate  the cases of respectively large and small differences between the considered distributions in Fig.~\ref{FigApp}b. Notice that the approximate distributions indeed approach the exact ones, even with greater accuracy for smaller $d_B$ as expected. Finally, since $\bar{d}_B=0.010$ is a very pessimistic assumption as it is only reachable for few values of $\phi$, results in  Fig.~\ref{FigApp}b validate the accuracy of \eqref{aa-is2} in general scenarios.

\bibliographystyle{IEEEtran}
\bibliography{IEEEabrv,references}
\end{document}